\documentclass[11pt, letterpaper]{article}
\usepackage{times}

\usepackage{amsmath,amsfonts,amssymb,amsthm}
\usepackage[colorlinks=true,linkcolor=red,citecolor=blue,urlcolor=red]{hyperref}
\usepackage{framed}
\usepackage{verbatim}
\usepackage[T1]{fontenc}
\usepackage{paralist}
\usepackage{stmaryrd}

\usepackage{mathptmx}
\usepackage[text={6.5in,9in}]{geometry}

\newenvironment{reminder}[1]{\medskip
\noindent {\bf Reminder of #1.\  }\em}{\smallskip}

\newtheorem{fact}{Fact}[section]

\newtheorem{theorem}{Theorem}[section]

\newtheorem{remark}{Remark}

\newtheorem{corollary}{Corollary}[section]
\newtheorem{lemma}{Lemma}[section]

\newtheorem{definition}{Definition}[section]

\newenvironment{proofof}[1]{\medskip
\noindent {\bf Proof of #1.  }}{\hfill$\Box$
\medskip}

\def \SUM {{\sf SUM}}
\def \MOD {{\sf MOD}}

\def \ACC {{\sf ACC}}
\def \XOR {{\sf XOR}}
\def \OR {{\sf OR}}

\def \TC {{\sf TC}}
\def \AC {{\sf AC}}

\def \MAJ {{\sf MAJ}}
\def \AND {{\sf AND}}
\def \THR {{\sf LTF}}
\def \LTF {\THR}
\def \NP {{\sf NP}}
\def \SYM {{\sf SYM}}

\def\eps{\varepsilon}
\def\poly{\text{poly}}
\def \Z {{\mathbb Z}}
\def \R {{\mathbb R}}
\def \F {{\mathbb F}}

\def\polylog{\operatorname{polylog}}

\newcommand{\D}{\Delta}
\newcommand{\down}[1]{\left\lfloor#1\right\rfloor}
\newcommand{\up}[1]{\left\lceil#1\right\rceil}
\newcommand{\predicate}[1]{\left[ #1 \right]}
\newcommand{\OO}{\tilde{O}}
\newcommand{\IGNORE}[1]{}
\def \E {\mathbb E}
\newcommand{\brac}[1]{\left[ #1 \right]}
\def \TH {{\rm TH}}
\newcommand{\PP}{\widetilde{P}}

\begin{document}
\title{Polynomial Representations of Threshold Functions and Algorithmic Applications}

\date{}
\author{
Josh Alman\footnote{Computer Science Department, Stanford University, {\tt jalman@cs.stanford.edu}. Supported by NSF CCF-1212372 and NSF DGE-114747.} \and
Timothy M. Chan\footnote{Cheriton School of Computer Science, University of Waterloo, {\tt tmchan@uwaterloo.ca}.  Supported by an NSERC grant.} \and
Ryan Williams\footnote{Computer Science Department, Stanford University, {\tt rrw@cs.stanford.edu}. Supported in part by NSF CCF-1212372 and CCF-1552651 (CAREER). Any opinions, findings and conclusions or recommendations expressed in this material are those of the authors and do not necessarily reflect the views of the National Science Foundation.} 
}

\maketitle

\begin{abstract}
We design new polynomials for representing threshold functions in three different regimes: \emph{probabilistic polynomials} of low degree, which need far less randomness than previous constructions, \emph{polynomial threshold functions} (PTFs) with ``nice'' threshold behavior and degree almost as low as the probabilistic polynomials, and a new notion of \emph{probabilistic PTFs} where we combine the above techniques to achieve even lower degree with similar ``nice'' threshold behavior. Utilizing these polynomial constructions, we design faster algorithms for a variety of problems: 

\begin{itemize} 
    \item {\bf Offline Hamming Nearest (and Furthest) Neighbors:} Given $n$ red and $n$ blue points in $d$-dimensional Hamming space for $d = c \log n$, we can find an (exact) nearest (or furthest) blue neighbor for every red point in randomized time $n^{2 - 1/O(\sqrt{c}\log^{2/3} c)}$ or deterministic time $n^{2 - 1/O(c\log^{2} c)}$. These improve on a randomized $n^{2 - 1/O(c\log^{2} c)}$ bound by Alman and Williams (FOCS'15), and also lead to faster MAX-SAT algorithms for sparse CNFs.
    
    \item {\bf Offline Approximate Nearest (and Furthest) Neighbors:} Given $n$ red and $n$ blue points in $d$-dimensional $\ell_1$ or Euclidean space, we can find a $(1+\eps)$-approximate nearest (or furthest) blue neighbor for each red point in randomized time near $dn+n^{2-\Omega(\eps^{1/3}/\log(1/\eps))}$. This improves on an algorithm by Valiant (FOCS'12) with randomized time near $dn+n^{2 - \Omega(\sqrt{\epsilon})}$, which in turn improves previous methods based on locality-sensitive hashing.


    \item {\bf SAT Algorithms and Lower Bounds for Circuits With Linear Threshold Functions:} We give a satisfiability algorithm for $\AC^0[m]\circ \LTF\circ \LTF$ circuits with a subquadratic number of linear threshold gates on the bottom layer, and a subexponential number of gates on the other layers, that runs in deterministic $2^{n-n^{\eps}}$ time. This strictly generalizes a SAT algorithm for $\ACC^0 \circ \LTF$ circuits of subexponential size by Williams (STOC'14) and also implies new circuit lower bounds for threshold circuits, improving a recent gate lower bound of Kane and Williams (STOC'16). We also give a randomized $2^{n - n^\eps}$-time SAT algorithm for subexponential-size  $\MAJ \circ \AC^0 \circ \LTF \circ \AC^0 \circ \LTF$ circuits, where the top $\MAJ$ gate and middle $\LTF$ gates have $O(n^{6/5-\delta})$ fan-in. 
\end{itemize}

\end{abstract}

\thispagestyle{empty}
\newpage
\setcounter{page}{1}

\section{Introduction}

The polynomial method is a powerful tool in circuit complexity. The idea of the method is to transform all circuits of some class into ``nice'' polynomials which represent the circuit in some way. If the polynomial is always sufficiently nice (e.g. has low degree), and one can prove that a certain Boolean function $f$ cannot be represented so nicely, one concludes that the circuit class is unable to compute $f$.

Recently, these tools have found surprising uses in algorithm design. If a subproblem of an algorithmic problem can be modeled by a simple circuit, and that circuit can be transformed into a ``nice'' polynomial (or ``nice'' distribution of polynomials), then fast algebraic algorithms can be applied to evaluate or manipulate the polynomial quickly. This approach has led to advances on problems such as All-Pairs Shortest Paths~\cite{WilliamsAPSP14}, Orthogonal Vectors and Constraint Satisfaction \cite{Williams-Yu14,AbboudWY15,WilliamsFSTTCS14}, All-Nearest Neighbor problems \cite{JoshRyan}, and Stable Matching~\cite{MoellerPaturiSchneider}.

In most applications, the key step is to randomly convert simple circuits into so-called \emph{probabilistic} polynomials. If $f$ is a Boolean function on $n$ variables, and $R$ is a ring, a \emph{probabilistic polynomial over $R$ for $f$ with error $1/s$ and degree $d$} is a distribution ${\cal D}$ of degree-$d$ polynomials over $R$ such that for all $x \in \{0,1\}^n$, $\Pr_{p \sim {\cal D}}[p(x) = f(x)] \geq 1-\frac{1}{s}$. Razborov~\cite{Razborov} and Smolensky~\cite{Smolensky87} introduced the notion of a probabilistic polynomial, and showed that any low-depth circuit consisting of AND, OR, and PARITY gates can be transformed into a low degree probabilistic polynomial by constructing constant degree probabilistic polynomials for those three gates. Many polynomial method algorithms use this transformation.

In this work, we are interested in polynomial representations of threshold functions. The threshold function $\TH_\theta$ determines whether at least a $\theta$ fraction of its input bits are 1s. Threshold functions are among the simplest Boolean functions that do not have constant degree probabilistic polynomials: Razborov and Smolensky showed that the MAJORITY function (a special case of a threshold function) requires degree $\Omega(\sqrt{n \log s})$. Nonetheless, as we will see throughout this paper, there are many important problems which can be reduced to evaluating circuits involving threshold gates on many inputs, and so further study of polynomial representations of threshold functions is warranted.

Threshold functions have been extensively studied in theoretical computer science for many years; there are numerous applications of linear and polynomial threshold functions to complexity and learning theory (a sample includes ~\cite{BeigelRS91,BruckS92,Aspnes94,Beigel95thepolynomial,Klivans-Servedio01,ODonnellS10,Sherstov14}). 

\subsection{Our Results}

We consider three different notions of polynomials representing $\TH_\theta$. Each achieves different tradeoffs between polynomial degree, the randomness required, and how accurately the polynomial represents $\TH_\theta$. Each leads to improved algorithms in our applications.

\textbf{Less Randomness.} First, we revisit probabilistic polynomials. Alman and Williams \cite{JoshRyan} designed a probabilistic polynomial for $\TH_\theta$ which already achieves a tight degree bound of $\Theta(\sqrt{n \log s})$. However, their construction uses $\Omega(n)$ random bits, which makes it difficult to apply in deterministic algorithms. We show how their low-degree probabilistic polynomials for threshold functions can use substantially fewer random bits:

\begin{theorem} \label{derand}
For any $0 \leq \theta \leq 1$, there is a probabilistic polynomial for the function $\TH_\theta$ of degree $O(\sqrt{n \log s})$ on $n$ bits with error $1/s$ that can be randomly sampled using only $O(\log n \log(ns))$ random bits.
\end{theorem}

\textbf{Polynomial Threshold Function Representations.} Second, we consider deterministic Polynomial Threshold Functions (PTFs). A PTF for a Boolean function $f$ is a polynomial (\emph{not} a distribution on polynomials) $p : \{0,1\}^n \to \R$ such that $p(x)$ is smaller than a fixed value when $f(x) = 0$, and $p(x)$ is larger than the value when $f(x)=1$.  In our applications, we seek PTFs with ``good threshold behavior'', such that $|p(x)|\le 1$ when $f(x)=0$, and $p(x)$ is very large otherwise. We can achieve almost the same degree for a PTF as for a probabilistic polynomial, and even better degree for an approximate threshold function:

\begin{theorem}\label{thm:det:poly}
We can construct a polynomial $P_{s,t,\eps}:\R\rightarrow\R$ of degree $O(\sqrt{1/\eps}\log s)$, such that
\begin{itemize}
\item if $x\in\{0,1,\ldots,t\}$, then $|P_{s,t,\eps}(x)|\le 1$;
\item if $x\in (t,(1+\eps)t)$, then $P_{s,t,\eps}(x)>1$;
\item if $x\ge (1+\eps)t$, then $P_{s,t,\eps}(x)\ge s$.
\end{itemize}
For the ``exact'' setting with $\eps=1/t$, we can alternatively bound the degree by $O(\sqrt{t\log(st)})$.
\end{theorem}

By summing multiple copies of the polynomial from Theorem~\ref{thm:det:poly}, we immediately obtain a PTF with the same degree for the OR of $O(s)$ threshold functions (needed in our applications). This theorem follows directly from known extremal properties of Chebyshev polynomials, as well as the lesser known \emph{discrete} Chebyshev polynomials. Because Theorem~\ref{thm:det:poly} gives a single polynomial instead of a distribution on polynomials, it is especially helpful for designing deterministic algorithms. Chebyshev polynomials are well-known to yield good approximate polynomials for computing certain Boolean functions over the reals~\cite{Nisan-Szegedy94,Paturi92,Klivans-Servedio01,Sherstov13,GregLightbulb} (please see the Preliminaries for more background). 


\textbf{Probabilistic PTFs.} Third, we introduce a new (natural) notion of a \emph{probabilistic PTF} for a Boolean function $f$. This is a distribution on PTFs, where for each input $x$, a PTF drawn from the distribution is highly likely to agree with $f$ on $x$. Combining the techniques from probabilistic polynomials for $\TH_\theta$ and the deterministic PTFs in a simple way, we construct a probabilistic PTF with good threshold behavior whose degree is \emph{lower} than both the deterministic PTF and the degree bounds attainable by probabilistic polynomials (surprisingly breaking the ``square-root barrier''):

\begin{theorem}\label{thm:poly}
We can construct a probabilistic polynomial $\PP_{n,s,t,\eps}:\{0,1\}^n\rightarrow\R$ of degree $O((1/\eps)^{1/3}\log s)$, such that
\begin{itemize}
\item if $\sum_{i=1}^n x_i \le t$, then $|\PP_{n,s,t,\eps}(x_1,\ldots,x_n)|\le 1$ with probability at least $1-1/s$;
\item if $\sum_{i=1}^n x_i \in (t,t+\eps n)$, then $\PP_{n,s,t,\eps}(x_1,\ldots,x_n)> 1$ with probability at least $1-1/s$;
\item if $\sum_{i=1}^n x_i\ge t+\eps n$, then $\PP_{n,s,t,\eps}(x_1,\ldots,x_n)\ge s$ with probability at least $1-1/s$.
\end{itemize}
For the ``exact'' setting with $\eps=1/n$, we can alternatively bound the degree by 
$O(n^{1/3}\log^{2/3}(ns))$.
\end{theorem}

The PTFs of Theorem~\ref{thm:poly} can be sampled using only $O(\log(n) \cdot \log(ns))$ random bits as well; their lower degree will allow us to design faster randomized algorithms for a variety of problems. For emphasis, we will sometimes refer to PTFs as \emph{deterministic PTFs} to distinguish them from probabilistic PTFs.

These polynomials for $\TH_\theta$ can be applied to many different problems:

\textbf{Offline Hamming Nearest Neighbor Search.} In the Hamming Nearest Neighbor problem, we wish to preprocess a set $D$ of $n$ points in $\{0,1\}^d$ such that, for a query $q \in \{0,1\}^d$, we can quickly find the $p \in D$ with smallest Hamming distance to $q$. This problem is central to many problems throughout Computer Science, especially in search and error correction \cite{IndykSurvey}. However, it suffers from the \emph{curse of dimensionality} phenomenon, where known algorithms achieve the nearly trivial runtimes of either $2^{\Omega(d)}$ or $\Omega(n/\poly(\log n))$, with matching lower bounds in many data structure models (see e.g. \cite{BarkolRabani}). Using our PTFs, we instead design a new algorithm for the natural offline version of this problem:

\begin{theorem}\label{thm:ham}
Given $n$ red and $n$ blue points in $\{0,1\}^d$ for $d=c\log n \ll \log^3 n/\log^5\log n$, we can find an (exact) Hamming nearest/farthest blue neighbor for every red point in randomized time
$n^{2-1/O(\sqrt{c}\log^{3/2}c)}$.
\end{theorem}

Using the same ideas, we are also able to derandomize our algorithm, to achieve \emph{deterministic} time $n^{2-1/O(c\log^{2}c)}$ (see Remark \ref{rmk:det:alg} in Section \ref{hnnsection}). When $d = c \log n$ for constant $c$, these algorithms both have ``truly subquadratic'' runtimes. These both improve on Alman and Williams' algorithm \cite{JoshRyan} which runs in randomized time $n^{2-1/O(c\log^{2}c)}$, and only gives a nontrivial algorithm for $d \ll \log^2 n/\log^3 \log n$. Applying reductions from \cite{JoshRyan}, we can achieve similar runtimes for finding closest pairs in $\ell_1$ for vectors with small integer entries, and pairs with maximum inner product or Jaccard coefficient.

It is worth noting that there may be a serious limit to solving this problem much faster. Theorem \ref{thm:ham} (and \cite{JoshRyan}) shows for all $c$ there is a $\delta>0$ such that Offline Hamming Nearest Neighbor search in dimension $d = c \log n$ takes $O(n^{2-\delta})$ time. Showing that there is a universal $\delta>0$ that works for all $c$ would disprove the Strong Exponential Time Hypothesis \cite[Theorem 1.4]{JoshRyan}.

\textbf{Offline Approximate Nearest Neighbor Search.}
The problem of finding high-dimensional \emph{approximate} nearest neighbors has received even more attention.  Locality-sensitive hashing yields data structures that can find $(1+\eps)$-factor approximate nearest neighbors to any query point in $\OO(dn^{1-\Omega(\eps)})$ (randomized) time after preprocessing in  $\OO(dn+n^{2-\Omega(\eps)})$ time and space,\footnote{Throughout the paper, the $\OO$ notation hides polylogarithmic factors, $[U]$ denotes $\{0,1,\ldots,U-1\}$, and $\poly(n)$ denotes a fixed polynomial in $n$.}
for not only Hamming space but also $\ell_1$ and $\ell_2$ space~\cite{HarPeledIndykMotwani,AndoniIndyk}.  Thus, a batch of $n$ queries can be answered in $\OO(dn^{2-\Omega(\eps)})$ randomized time.  Exciting recent work on locality-sensitive hashing \cite{AndoniBeyond,andoni2015optimal} has improved the constant factor in the $\Omega(\eps)$ bound, but not the growth rate in $\eps$.  In 2012, G. Valiant \cite{GregLightbulb} reported a surprising algorithm running in $\OO(dn + n^{2-\Omega(\sqrt{\eps})})$ randomized time for the offline version of the problem in $\ell_2$.
We obtain a still faster algorithm for the offline problem, with $\sqrt{\eps}$ improved to about $\eps^{1/3}$:

\begin{theorem}\label{thm:ham:approx}
Given $n$ red and $n$ blue points in $[U]^d$ and $\eps\gg \frac{\log^6\log n}{\log^3 n}$, we can find 
a $(1+\eps)$-approximate $\ell_1$ or $\ell_2$ nearest/farthest blue neighbor for each red point in 
$(dn+n^{2-\Omega(\eps^{1/3}/\log(1/\eps))})\cdot \poly(\log(nU))$ randomized time.
\end{theorem}

Valiant's algorithm, like Alman and Williams'~\cite{JoshRyan}, relied on fast matrix multiplication, and it also used Chebyshev polynomials but in a seemingly more complicated way.  Our new probabilistic PTF construction is inspired by our attempt to unify Valiant's approach with Alman and Williams', which leads to not only a simplification but also an improvement of Valiant's algorithm.  (We also almost succeed in derandomizing Valiant's $n^{2-\tilde{\Omega}(\sqrt{\eps})}$ result in the Hamming case, except for an initial dimension reduction step;
see Remark \ref{rmk:det:alg} in Section \ref{hnnsection}.)

Numerous applications to high-dimensional computational geometry follow; for example, we can approximate the diameter or Euclidean minimum spanning tree in roughly the same running time.

\textbf{MAX-SAT.} Another application is MAX-SAT: finding an assignment that satisfies the maximum number of clauses in a given CNF formula with $n$ variables.  In the sparse case when the number of clauses is $cn$, a series of papers have given faster exact 
algorithms, for example, achieving
$2^{n-n/O(c\log c)}$ time by Dantsin and Wolpert~\cite{DantsinWolpert},
$2^{n-n/O(c\log c)^{2/3}}$ time by Sakai et al.~\cite{SakSetTamTer}, and
$2^{n-n/O(\sqrt{c})}$ time by Chen and Santhanam~\cite{ChenSanthanam}.  Using the polynomial method and our new probabilistic PTF construction, we obtain the following improved result:

\begin{theorem}\label{thm:maxsat}
Given a CNF formula with $n$ variables and $cn\ll n^4/\log^{10} n$ clauses,
we can find an assignment that satisfies the maximum number of clauses
in randomized $2^{n-n/O(c^{1/3}\log^{7/3}c)}$ time.
\end{theorem}

For general dense instances, the problem becomes tougher. Williams~\cite{Williams05} gave an $O(2^{0.792n})$-time algorithm for MAX-2-SAT, but an $O(2^{(1-\delta)n})$-time algorithm for MAX-3-SAT (for a universal $\delta > 0$) has remained open; currently the best reported time bound~\cite{SakSetTamTerECCC15} is $2^{n-\Omega(n/\log n)^{1/3}}$, which can be slightly improved to $2^{n-\Omega(\sqrt{n/\log n})}$ with more care. We make new progress on not only MAX-3-SAT but also MAX-4-SAT:

\begin{theorem}\label{thm:max4sat}
Given a weighted 4-CNF formula $F$ with $n$ variables with positive integer weights bounded by $\poly(n)$, we can find an assignment that maximizes the total weight of clauses satisfied in $F$, in randomized $2^{n-n/O(\log^2 n\log^2\log n)}$ time.
In the sparse case when the clauses have total weight $cn$, the time bound improves
to $2^{n-n/O(\log^2 c\log^2\log c)}$.
\end{theorem}

\textbf{LTF-LTF Circuit SAT Algorithms and Lower Bounds.} Using our small sample space for probabilistic MAJORITY polynomials  (Theorem~\ref{derand}), we construct a new circuit satifiability algorithm for circuits with linear threshold functions (LTFs) which improves over several prior results. Let $\AC^0[d,m]\circ \LTF\circ \LTF[S_1,S_2,S_3]$ be the class of circuits with a layer of $S_3$ LTFs at the bottom layer (nearest the inputs), a layer of $S_2$ LTFs above the bottom layer, and a size-$S_1$ $\AC^0[m]$ circuit of depth $d$ above the two LTF layers.\footnote{Recall that for an integer $m \geq 2$, $\AC^0[m]$ refers to constant-depth unbounded fan-in circuits over the basis $\{\AND,\OR,\MOD_m\}$, where $\MOD_m$ outputs $1$ iff the sum of its input bits is divisible by $m$.} 

\begin{theorem} \label{ACC-LTF-SAT} 
For every integer $d > 0$, $m > 1$, and $\delta > 0$, there is an $\eps > 0$ and an algorithm for satisfiability of $\AC^0[d,m]\circ \LTF\circ \LTF[2^{n^{\eps}},2^{n^{\eps}},n^{2-\delta}]$ circuits that runs in deterministic $2^{n-n^{\eps}}$ time.
\end{theorem}

Williams~\cite{WilliamsACCTHR14} gave a comparable SAT algorithm for $\ACC^0 \circ \THR$ circuits of $2^{n^{\eps}}$ size, where $\eps > 0$ is sufficiently small.\footnote{Recall $\ACC^0$ is the infinite union of $\AC^0[m]$ for all integers $m \geq 2$.} Theorem~\ref{ACC-LTF-SAT} strictly generalizes the previous algorithm, allowing another layer of $n^{2-\eps}$ linear threshold functions below the existing $\THR$ layer. Theorem~\ref{ACC-LTF-SAT} also trivially implies deterministic SAT algorithms for $\THR \circ \THR$ circuits of up to $n^{2-o(1)}$ gates, improving over the recent SAT algorithms of Chen, Santhanam, and Srinivasan~\cite{ChenSS15} which only work for $n^{1+\eps}$-wire circuits for $\eps \ll 1$, and the SAT algorithms of Impagliazzo, Paturi, and Schneider~\cite{IPS13}.

Here we sketch the ideas in the SAT algorithm for $\ACC^0 \circ \LTF \circ \LTF$. Similar to the SAT algorithm for $\ACC^0 \circ \LTF$ circuits~\cite{WilliamsACCTHR14}, the bottom layer of $\LTF$s can be replaced by a layer of DNFs, via a weight reduction trick. 
We replace $\LTF$s in the middle layer with $\AC^0 \circ \MAJ$ circuits (modifying a construction of Maciel and Th\'erien~\cite{Maciel-Therien98} to keep the fan-in of $\MAJ$ gates low), then replace these $\MAJ$ gates of $n^{2-\Theta(\delta)}$ fan-in with probabilistic $\F_2$-polynomials of degree $n^{1-\Theta(\delta)+\Theta(\eps)}$ over a small sample space, provided by Theorem~\ref{derand}. Taking a majority vote over all samples, and observing that an $\F_2$-polynomial is a $\MOD_2 \circ \AND$ circuit, we obtain a $\MAJ \circ \ACC^0$ circuit, but with $2^{n^{1-O(\delta)}}$ size in some of its layers. By carefully applying known depth reduction techniques, we can convert the circuit into a depth-two circuit of size $2^{n^{1-\Omega(\eps)}}$ which can then be evaluated efficiently on many inputs. (This is not obvious: applying the Beigel-Tarui depth reduction to a $2^{O(n^{1-\eps})}$-size circuit would make its new size  \emph{quasi-polynomial in $2^{O(n^{1-\eps})}$}, yielding an intractable bound of $2^{n^{O(1)}}$.) 

Applying the known connection between circuit satisfiability algorithms and circuit lower bounds for ${\sf E}^{\NP}$ problems~\cite{Williams10,WilliamsJACM14,JMV13}, the following is immediate:

\begin{corollary} \label{NEXP-lbs} For every $d > 0$, $m > 1$, and $\delta \in (0,1)$, there is an $\eps > 0$ such that the class ${\sf E}^{\NP}$ does not have non-uniform circuits in $\AC^0[d,m]\circ \LTF \circ \LTF[2^{n^{\eps}},2^{n^{\eps}},n^{2-\delta}]$. In particular, for every $\eps > 0$, ${\sf E}^{\NP}$ does not have $\ACC^0 \circ \LTF \circ \LTF$ circuits where the $\ACC^0 \circ \LTF$ subcircuit has $2^{n^{o(1)}}$ size and the bottom $\LTF$ layer has $n^{2-\eps}$ gates.
\end{corollary}

Most notably, Corollary~\ref{NEXP-lbs} proves lower bounds with $n^{2-\eps}$ LTFs on the bottom layer and \emph{subexponentially many} LTFs on the second layer. This improves upon recent $\LTF \circ \LTF$ gate lower bounds of Kane and Williams~\cite{KaneW15}, at the cost of raising the complexity of the hard function from $\TC^0_3$ to ${\sf E}^{\NP}$. Suguru Tamaki~\cite{Tamaki16} has recently reported similar results for depth-two circuits with both symmetric and threshold gates.

\textbf{A Powerful Randomized SAT Algorithm.} Finally, combining the probabilistic PTF for MAJORITY (Theorem~\ref{thm:poly}) with the probabilistic polynomial of \cite{JoshRyan}, we give a randomized SAT algorithm for a rather powerful class of circuits. The class $\MAJ \circ \AC^0 \circ \LTF \circ \AC^0 \circ \LTF$ denotes the class of circuits with a majority gate at the top, along with two layers of linear threshold gates, and arbitrary $O(1)$-depth $\AC^0$ circuitry between these three layers. This circuit class is arguably much more powerful than $\TC^0_3$ ($\MAJ \circ \MAJ \circ \MAJ$), based on known low-depth circuit constructions for arithmetic functions (e.g. \cite{Chandra-Stockmeyer-Vishkin84,Maciel-Therien98,MacielT99}). 

\begin{theorem} \label{TC03-SAT} 
For all $\eps > 0$ and integers $d \geq 1$, there is a $\delta > 0$ and a randomized satisfiability algorithm for $\MAJ \circ \AC^0 \circ \LTF \circ \AC^0 \circ \LTF$ circuits of depth $d$ running in $2^{n - \Omega(n^{\delta})}$ time, on circuits with the following properties:\begin{compactitem} 
\item the top $\MAJ$ gate, along with every $\LTF$ on the middle layer, has $O(n^{6/5-\eps})$ fan-in, and
\item there are $O(2^{n^{\delta}})$ many $\AND/\OR$ gates (anywhere) and $\LTF$ gates at the bottom layer. 
\end{compactitem} 
\end{theorem}

Theorem~\ref{TC03-SAT} applies the probabilistic PTF of degree about $n^{1/3}$ (Theorem~\ref{thm:poly}) to the top $\MAJ$ gate, probabilistic polynomials over $\Z$ of degree about $n^{1/2}$ (Theorem~\ref{derand}) to the middle LTFs, and  weight reduction 
to the bottom LTFs; the rest can be represented with $\poly(n^{\delta})$ degree.  

It would not be surprising (to at least one author) if the above circuit class contained strong pseudorandom function candidates; that is, it seems likely that the Natural Proofs barrier applies to this circuit class. Hence from the circuit lower bounds perspective, the problem of derandomizing the SAT algorithm of Theorem~\ref{TC03-SAT} is extremely interesting.


\section{Preliminaries}

\textbf{Notation.} In what follows, for $(x_1,\ldots,x_n) \in \{0,1\}^n$ define $|x| := \sum_{i=1}^n x_i$. For a logical predicate $P$, we use the notation $\brac{P}$ to denote the function which outputs $1$ when $P$ is true, and $0$ when $P$ is false.

For $\theta \in [0,1]$, define $\TH_\theta : \{ 0,1 \}^n \to \{ 0, 1 \}$ to be the \emph{threshold function} $\TH_\theta(x_1,\ldots,x_n) := \brac{|x|/n \geq \theta}$. In particular, $\TH_{1/2} = {\rm MAJORITY}$.

For classes of circuits ${\cal C}$ and ${\cal D}$, ${\cal C} \circ {\cal D}$ denotes the class of circuits consisting of a single circuit $C \in {\cal C}$ whose inputs are the outputs of some circuits from ${\cal D}$. That is, ${\cal C} \circ {\cal D}$ is simply the composition of circuits from ${\cal C}$ and ${\cal D}$.

\textbf{Rectangular Matrix Multiplication.} One of our key tools is fast rectangular matrix multiplication:

\begin{lemma}[Coppersmith~\cite{Coppersmith82}]\label{rectangular} For all sufficiently large $N$, multiplication of an $N \times N^{.172}$ matrix with an $N^{.172} \times N$ matrix can be done in $O(N^2 \log^2 N)$ arithmetic operations over any field.\end{lemma}

A proof can be found in the appendix of~\cite{WilliamsACCTHR14}.

\textbf{Chebyshev Polynomials in TCS.} Another key to our work is that we find new applications of Chebyshev polynomials to algorithm design. This is certainly not a new phenomenon in itself; here we briefly survey some prior related usages of Chebyshev polynomials. First, Nisan and Szegedy~\cite{Nisan-Szegedy94} used Chebyshev polynomials to compute the OR function on $n$ Boolean variables with an ``approximating'' polynomial $p: \R^n \rightarrow \R$, such that for all $x \in \{0,1\}^n$ we have $|OR(x) - p(x)| \leq 1/3$, yet $\deg(p) = O(\sqrt{n})$. They also proved the degree bound is tight up to constants in the big-O; Paturi~\cite{Paturi92} generalized the upper and lower bound to all symmetric functions.

This work has led to several advances in learning theory. Building on the polynomials of Nisan and Szegedy, Klivans and Servedio~\cite{Klivans-Servedio01} showed how to compute an OR of $t$ ANDs of $w$ variables with a PTF of degree $O(\sqrt{w} \log t)$, similar to our degree bound for computing an OR of $t$ MAJORITYs of $w$ variables of Theorem~\ref{thm:det:poly} (however, note our bound in the ``exact'' setting is a bit better, due to our use of discrete Chebyshev polynomials). They also show how to compute an OR of $s$ ANDs on $n$ variables with a \emph{deterministic} PTF of $O(n^{1/3} \log s)$ degree, similar to our cube-root-degree probabilistic PTF for the OR of MAJORITY of Theorem~\ref{thm:poly} in the ``exact'' setting. However, it looks difficult to generalize Klivans-Servedio's $O(n^{1/3} \log s)$ degree bound to compute an OR of MAJORITY: part of their construction uses a reduction to decision lists which works for conjunctions but not for MAJORITY functions. Klivans, O'Donnell and Servedio~\cite{KlivansOS04} show how to compute an AND of $k$ MAJORITY on $n$ variables with a PTF of degree $O(\sqrt{w} \log k)$. By a simple transformation via De Morgan's law, there is a polynomial for OR of MAJORITY with the same degree. Their degree is only slightly worse than ours in terms of $k$ (because we use discrete Chebyshev polynomials).

In streaming algorithms, Harvey, Nelson, and Onak~\cite{HarveyNO08} use Chebyshev polynomials to design efficient algorithms for computing various notions of entropy in a stream. As a consequence of a query upper bound in quantum computing, Ambainis et al.~\cite{Ambainis2010any} show how to approximate any Boolean formula of size $s$ with a polynomial of degree $\sqrt{s}^{1+o(1)}$, improving on earlier bounds of O'Donnell and Servedio~\cite{ODonnellS10} that use Chebyshev polynomials. Sachdeva and Vishnoi~\cite{sachdeva2013approximation} give applications of Chebyshev polynomials to graph algorithms and matrix algebra. Linial and Nisan \cite{linial1990approximate} use Chebyshev polynomials to approximate inclusion-exclusion formulas, and Sherstov \cite{sherstov2008approximate} extends this to arbitrary symmetric functions.

\section{Derandomizing Probabilistic Polynomials for Threshold Functions}

In this section, we revisit the previous probabilistic polynomial for the majority function on $n$ bits, and show it can be implemented using only $\polylog(n, s)$ random bits. Our construction is essentially identical to that of \cite{JoshRyan}, except that we use far fewer random bits to sample entries from the input vector in the recursive step of the construction. 

For the analysis, we need a Chernoff bound for bits with limited independence:

\begin{lemma}[\cite{SSS95} Theorem 5 (I)(b)] \label{SSS}
If $X$ is the sum of $k$-wise independent random variables, each of which is confined to the interval $[0,1]$, with $\mu = \E[X]$, $\delta \leq 1$, and $k = \lfloor \delta^2 \mu e^{-1/3} \rfloor$, then
$$\Pr[|X - \mu| \geq \delta \mu] \leq e^{- \delta^2 \mu / 3}.$$
\end{lemma}

In particular, the following inequality appears in the analysis of \cite{JoshRyan}:

\begin{corollary} \label{chernoff}
If $x \in \{ 0,1 \}^n$ with $|x|/n = w$, and $\tilde{x} \in \{0,1\}^{n/10}$ is a vector each of whose entries is $k$-wise independently chosen entry of $x$, where $k = \lfloor 20 e^{-1/3} \log(1/\epsilon) \rfloor$, with $|\tilde{x}|/(n/10) = v$, then for every $\eps < 1/4$,
\[\Pr\left[v \leq w - \frac{a}{\sqrt{n}}\right] \leq \frac{\eps}{4},\]
where $a = \sqrt{10} \cdot \sqrt{\ln (1/\epsilon)}$.
\end{corollary}

\begin{proof}
Apply Lemma \ref{SSS} with $X = |\tilde{x}|$, $\mu = \E[|\tilde{x}|] = wn$, and $\delta = \sqrt{40 \log(1/\epsilon)/n}$.
\end{proof}

\begin{reminder}{Theorem~\ref{derand}} 
For any $0 \leq \theta \leq 1$, there is a probabilistic polynomial for the threshold function $\TH_\theta$ of degree $O(\sqrt{n \log s})$ on $n$ bits with error $1/s$ that can be randomly sampled using $O(\log(n) \log(ns))$ random bits.
\end{reminder}

\begin{proof}
Our polynomial is defined recursively, just as in \cite{JoshRyan}. Set $\epsilon = 1/s$. Using their notation, the polynomial $M_{n, \theta, \epsilon}$ for computing $\TH_\theta$ on $n$ bits with error $\epsilon$ is defined by:
\[M_{n, \theta, \epsilon}(x) := A_{n, \theta, 2a}(x) \cdot S_{n/10, \theta, a/\sqrt{n}, \epsilon/4}(\tilde{x})  +  M_{n/10, \theta, \epsilon/4}(\tilde{x}) \cdot (1 - S_{n/10, \theta, a/\sqrt{n}, \epsilon/4}(\tilde{x})).\]

In~\cite{JoshRyan}, $\tilde{x}$ was a sample of $n/10$ bits of $x$, chosen independently at random. Here, we pick $\tilde{x}$ to be a sample of $n/10$ bits chosen $k$-wise independently, for $k = \lfloor 20 e^{-1/3} \log(1/\epsilon) \rfloor$. The other polynomials in this recursive definition are as in \cite{JoshRyan}:

\begin{compactitem}
\item $M_{m, \theta, \epsilon}$ for $m<n$ is the (recursively defined) probabilistic polynomial for $\TH_\theta$ on $m$ bits and $\epsilon$ error
\item $S_{m, \theta, \delta, \epsilon}(x) := (1-M_{m, \theta + \delta, \epsilon}(x)) \cdot M_{m, \theta - \delta, \epsilon}(x)$ for $m<n$
\item $A_{n, \theta, g} : \{0,1\}^n \to \Z$ is an exact polynomial of degree at most $2 g \sqrt{n} + 1$ which gives the correct answer to $\TH_\theta$ for any vector $x$ with $|x| \in [\theta n - g \sqrt{n}, \theta n + g \sqrt{n}]$, and may give arbitrary answers on other vectors.
\end{compactitem}

Examining the proof of correctness in Alman and Williams~\cite{JoshRyan}, we see that the only requirement of the randomness is that it satisfies their Lemma 3.4, a concentration inequality for sampling $\tilde{x}$ from $x$. Our Corollary \ref{chernoff} is identical to their Lemma 3.4, except that it replaces their method of sampling $\tilde{x}$ with $k$-wise sampling; the remainder of the proof of correctness is exactly as before.

Our polynomial construction is recursive: we divide $n$ by $10$ and divide $\epsilon$ by $4$, each time we move from one recursive layer to the next. At the $j$th recursive level of our construction, for $1 \leq j < \log_{10}(n)$, we need to $O(\log(4^j/\epsilon))$-wise independently sample $n/10^j$ entries from a vector of length $n/10^{j-1}$. Summing across all of the layers, we need a total of $O(n)$ samples from a $k$-wise independent space, where $k$ is never more than $O(n/\epsilon)$. This can be done all together using $O(n)$ samples from $\{1,2,\ldots,n\}$ which are $O(n/\epsilon)$-wise independent. Using standard constructions, this requires $O(\log(n) \log(n/\epsilon))$ random bits.
\end{proof}

\section{PTFs for ORs of Threshold Functions}

In this section, we show how to construct low-degree PTFs representing threshold functions that have good threshold behavior, and consequently obtain low-degree PTFs for an OR of many threshold functions. 

\subsection{Deterministic Construction}

We begin by reviewing some basic facts about Chebyshev polynomials. The \emph{degree-$q$ Chebyshev polynomial of the first kind} is \[T_q(x) := \sum_{i=0}^{\down{q/2}} {\binom{q}{2i}} (x^2-1)^i x^{q-2i}.\] 

\begin{fact}\label{fact:cheby} For any $\eps\in (0,1)$,
\begin{itemize}
\item if $x\in [-1,1]$, then $|T_q(x)|\le 1$;
\item if $x\in (1,1+\eps)$, then $T_q(x)>1$;
\item if $x\ge 1+\eps$, then $T_q(x)\ge \frac{1}{2} e^{q\sqrt{\eps}}$.
\end{itemize}
\end{fact}

\begin{proof}
The first property easily follows from the known formula
$T_q(x)=\cos(q\arccos(x))$ for $x\in [-1,1]$. The second and third properties follow from another known formula 
$T_q(x)=\cosh(q\operatorname{arcosh}(x))$ for $x>1$, which for $x\ge 1+\eps$ implies $T_q(x)\ge \cosh(q\sqrt{\eps})= \frac{1}{2}(e^{q\sqrt{\eps}}+e^{-q\sqrt{\eps}})$.
\end{proof}


In certain scenarios, we obtain slightly better results using a (lesser known) family of \emph{discrete Chebyshev polynomials} defined as follows~\cite[page 59]{HirvensaloThesis}:
$$D_{q,t}(x) := \sum_{i=0}^q (-1)^i {\binom{q}{i}}
{\binom{t-x}{q-i}} {\binom{x}{i}}.$$
(See also \cite[pages 33--34]{SzegoBook} or
Chebyshev's original paper~\cite{Chebyshev} 
with an essentially equivalent definition up to rescaling.)

\begin{fact}\label{fact:discrete:cheby}
Let $c_{q,t}=(t+1)^{q+1}/q!$.  For all $t> q\ge \sqrt{8(t+1)\ln(t+1)}$,
\begin{itemize}
\item if $x\in\{0,1,\ldots,t\}$, then
$|D_{q,t}(x)|\le c_{q,t}$;
\item if $x\le -1$, then
$D_{q,t}(x)\ge e^{q^2/(8(t+1))}c_{q,t}$.
\end{itemize}
\end{fact}

\begin{proof}
From \cite[page 61]{HirvensaloThesis},
\begin{eqnarray*}
\sum_{k=0}^t D_{q,t}(k)^2 &=& \binom{2q}{q}\binom{t+1+q}{2q+1}\\
&=& \frac{2q(2q-1)\cdots q}{q(q-1)\cdots 1}\cdot
    \frac{(t+1+q)(t+q)\cdots (t+1-q)}{(2q+1)(2q)\cdots 1}\\
&=& \frac{(t+1)((t+1)^2-1^2)((t+1)^2-2^2)\cdots ((t+1)^2-q^2)}
{(2q+1)(q!)^2}
\ \le\ \frac{(t+1)^{2q+2}}{(q!)^2}.
\end{eqnarray*}
Thus, for every integer $x\in [0,t]$, we have
$|D_{q,t}(x)|\le (t+1)^{q+1}/q! = c_{q,t}$.

For $x\le -1$, we have $(-1)^i \binom{x}{i} = 
\frac{(-x)(-x+1)\cdots (-x+i-1)}{1\cdot 2\cdots i}\ge 1$,
and by the Chu--Vandermonde identity,
\begin{eqnarray*}
D_{q,t}(x) &\ge& \sum_{i=0}^q \binom{q}{i}
\binom{t+1}{q-i}
\ =\ \binom{t+1+q}{q}\\
&=& \frac{(t+1)^q(1+\frac{1}{t+1})(1+\frac{2}{t+1})\cdots (1+\frac{q}{t+1})}{q!} \\
&\ge& \frac{c_{q,t}}{t+1} e^{\frac{1+2+\cdots+q}{2(t+1)}}\ =\ e^{q(q+1)/(4(t+1))-\ln(t+1)} c_{q,t}\ \ge\ e^{q^2/(8(t+1))}c_{q,t}.
\end{eqnarray*}

\vspace{-4ex}
\end{proof}

\begin{reminder}{Theorem~\ref{thm:det:poly}} 
We can construct a polynomial $P_{s,t,\eps}:\R\rightarrow\R$ of degree $O(\sqrt{1/\eps}\log s)$, such that
\begin{itemize}
\item if $x\in\{0,1,\ldots,t\}$, then $|P_{s,t,\eps}(x)|\le 1$;
\item if $x\in (t,(1+\eps)t)$, then $P_{s,t,\eps}(x)>1$;
\item if $x\ge (1+\eps)t$, then $P_{s,t,\eps}(x)\ge s$.
\end{itemize}
For the ``exact'' setting with $\eps=1/t$, we can alternatively bound the degree by $O(\sqrt{t\log(st)})$.
\end{reminder}
\begin{proof}
Set $P_{s,t,\eps}(x) := T_q(x/t)$ for a parameter~$q$ to be determined. 
The first two properties are obvious from Fact~\ref{fact:cheby}.  On the other hand, if $x\ge (1+\eps)t$, then Fact~\ref{fact:cheby} shows that
$P_{s,t,\eps}(x)\ge \frac{1}{2}e^{q\sqrt{\eps}}\ge s$, provided we set $q=\up{\sqrt{1/\eps}\ln(2s)}$.  This achieves $O(\sqrt{1/\eps}\log s)$ degree.

When $\eps = 1/t$ the above yields $O(\sqrt{t}\log s)$ degree; we can reduce the $\log s$ factor 
by instead defining $P_{s,t,\eps}(x) := D_{q,t}(t-x)/c_{q,t}$.
Now, if $x\ge t+1$, then
$P_{s,t,\eps}(x)\ge e^{q^2/(8(t+1))}\ge s$ by
setting $q=\up{\sqrt{8(t+1)\ln(\max\{s,t+1\})}}$.  
\IGNORE{
with the following modified construction: define
$P_{s,t,\D}(x) := \frac{(x-t)(x-(t-1))\cdots (x-(t-b+1))}{t(t-1)
\cdots(t-b+1)}\cdot T_q\left(\frac{x}{t-b}\right)$ for an appropriate choice of parameters $b$ and $q$.
For correctness, consider three cases:
\begin{itemize}
\item {\sc Case 1}: $x\le t-b$.  Clearly, $|P_{s,t,\eps}(x)|\le 1$.
\item {\sc Case 2}: $x\in [t-b+1,t]$.  Then since $x$ is an integer,
$P_{s,t,\eps}(x)=0$.
\item {\sc Case 3}: $x\ge t+1$. Then
$P_{s,t,\eps}(x)\ge \frac{b!}{t(t-1)
\cdots (t-b+1)} \cdot \frac{1}{2}e^{\sqrt{b/(t-b)}q} = \frac{1}{2}e^{\sqrt{b/(t-b)}q}/{t\choose b}
\ge s$, 
by setting $q = \sqrt{(t-b)/b}\ln \left(2s{t\choose b}\right)$.
\end{itemize}
The degree is
\[ O\left(b + \sqrt{t/b}(\log s + b\log(t/b))\right),
\]
which is $O(\sqrt{t\log s\log\frac{t}{\log s}})$
by setting $b = \up{\log s/\log \frac{t}{\log s}}$.
}
\end{proof} \smallskip

Using Theorem~\ref{thm:det:poly}, we can construct a low-degree PTF for computing an OR of $s$ thresholds of $n$ bits:

\begin{corollary}\label{cor:det:poly}
Given $n,s,t,\eps$, 
we can construct a polynomial $P:\{0,1\}^{ns}\rightarrow\R$ of degree at most $\D := O(\sqrt{1/\eps}\log s)$ and at most $s\cdot \binom{n}{\D}$ monomials, such that 
\begin{itemize}
\item if the formula $\bigvee_{i=1}^s \predicate{\sum_{j=1}^n x_{ij} > t}$ is false, then $|P(x_{11},\ldots,x_{1n},\ldots,x_{s1},\ldots,x_{sn})|\le s$;
\item if the formula $\bigvee_{i=1}^s \predicate{\sum_{j=1}^n x_{ij} \ge t+\eps n}$ is true, then $P(x_{11},\ldots,x_{1n},\ldots,x_{s1},\ldots,x_{sn})>2s$.
\end{itemize}
For the exact setting with $\eps=1/n$, we can alternatively bound $\D$ by
$O(\sqrt{n\log(ns)})$.
\end{corollary}
\begin{proof}
Define
$P(x_{11},\ldots,x_{1n},\ldots,x_{s1},\ldots,x_{sn})
\ :=\ \sum_{i=1}^s P_{n,3s,t,\eps}\left(\sum_{j=1}^n x_{ij}\right),$
where $P_{n,3s,t,\eps}$ is from Theorem~\ref{thm:det:poly}. The stated properties clearly hold. (In the second case, the output is at least $3s-(s-1)>2s$.)
\end{proof}

\subsection{Probabilistic Construction} \label{section:probPTF}

Allowing ourselves a distribution of PTFs to randomly draw from, we can achieve noticeably lower degree than the previous section. We start with a fact which follows easily from the (tight) probabilistic polynomial for MAJORITY:

\begin{fact}\label{fact:AlmWil}
{\bf (Alman--Williams~\cite{JoshRyan}, or Theorem \ref{derand})}
We can construct a probabilistic polynomial $Q_{n,s,t}:\{0,1\}^n\rightarrow\R$ of degree $O(\sqrt{n\log s})$, such that 
\begin{itemize}
\item if $\sum_{i=1}^n x_i \le t$, then $Q_{n,s,t}(x_1,\ldots,x_n)=0$ with probability at least $1-1/s$;
\item if $\sum_{i=1}^n x_i > t$, then $Q_{n,s,t}(x_1,\ldots,x_n)=1$ with probability at least $1-1/s$.
\end{itemize}
\end{fact}

\begin{reminder}{Theorem~\ref{thm:poly}} 
We can construct a probabilistic polynomial $\PP_{n,s,t,\eps}:\{0,1\}^n\rightarrow\R$ of degree $O((1/\eps)^{1/3}\log s)$, such that
\begin{itemize}
\item if $\sum_{i=1}^n x_i \le t$, then $|\PP_{n,s,t,\eps}(x_1,\ldots,x_n)|\le 1$ with probability at least $1-1/s$;
\item if $\sum_{i=1}^n x_i \in (t,t+\eps n)$, then $\PP_{n,s,t,\eps}(x_1,\ldots,x_n)> 1$ with probability at least $1-1/s$;
\item if $\sum_{i=1}^n x_i\ge t+\eps n$, then $\PP_{n,s,t,\eps}(x_1,\ldots,x_n)\ge s$ with probability at least $1-1/s$.
\end{itemize}
For the ``exact'' setting with $\eps=1/n$, we can alternatively bound the degree by
$O(n^{1/3}\log^{2/3}(ns))$.
\end{reminder}
\begin{proof}
Let $r$ and $q$ be parameters to be set later.
Draw a random sample $R\subseteq \{1,\ldots,n\}$ of size $r$.
Let $$t_R := \frac{tr}{n} - c_0\sqrt{r\log s} \quad\textrm{and}\quad
    t^- := t - 2c_0 \left(\frac{n}{\sqrt{r}}\right)\sqrt{\log s}$$
for a sufficiently large constant $c_0$.
Define 
\setlength{\thinmuskip}{0mu}

\[ \PP_{n,s,t,\eps}(x_1,\ldots,x_d)\ :=\ Q_{r,2s,t_R}(\{x_i\}_{i\in R})\:\cdot\:
                   P_{s,t',\eps'}\left(\sum_{i=1}^n x_i - t^-\right),
\]
\setlength{\thinmuskip}{3mu}

where $P_{s,t',\eps'}$ is the polynomial from Theorem \ref{thm:det:poly}, with $t':=t-t^-=\Theta((n/\sqrt{r})\sqrt{\log s})$ and
$\eps':=\eps n/t' = \Theta(\eps\sqrt{r}/\sqrt{\log s})$.

To verify the stated properties, consider three cases:
\begin{itemize}
\item
{\sc Case 1}: $\sum_{i=1}^n x_i < t^-$.
By a standard Chernoff bound, with probability at least $1-1/(2s)$,
we have $\sum_{i\in R} x_i < t^-r/n + c_0\sqrt{r\log s} \le t_R$ 
(assuming that $r\ge\log s$).
Thus, with probability at least $1-1/s$,
we have $Q_{n,2s,t_R}(\{x_i\}_{i\in R})=0$ and so
$\PP_{n,s,t,\eps}(x_1,\ldots,x_n)= 0$.
\item
{\sc Case 2}: $\sum_{i=1}^n x_i \in [t^-, t]$.
With probability at least $1-1/s$,
we have $Q_{r,2s,t_R}(\{x_i\}_{i\in R})\in\{0,1\}$ and so
$|\PP_{n,s,t,\eps}(x_1,\ldots,x_n)|\le 1$.
\item
{\sc Case 3}: $\sum_{i=1}^n x_i > t$.
By a standard Chernoff bound, with probability at least $1-1/(2s)$,
we have $\sum_{i\in R} x_i \ge tr/n +
c_0\sqrt{r\log s}=t_R$.
Thus, with probability at least $1-1/s$,
we have $Q_{r,2s,t_R}(\{x_i\}_{i\in R})=1$ and so
$\PP_{n,s,t,\eps}(x_1,\ldots,x_n)>1$ for $\sum_{i=1}^n x_i\in (t,t+\eps n)$, or
$\PP_{n,s,t,\eps}(x_1,\ldots,x_n)\ge s$ for $\sum_{i=1}^n x_i\ge t+\eps n$.
\end{itemize}
The degree of $\PP_{n,s,t,\eps}$ is
\[ O\left( \sqrt{r\log s} + \sqrt{(1/(\eps\sqrt{r}))\sqrt{\log s}}\log s\right) \]
and we can set $r = \up{(1/\eps)^{2/3} \log s}$.
For the exact setting,
the degree is
\[ O\left( \sqrt{r\log s} +
\sqrt{(n/\sqrt{r})\sqrt{\log s}\cdot \log(ns)} \right)
\]
and we can set
$r=\up{n^{2/3}\log^{1/3}(ns)}$.
\end{proof}

\begin{remark}\rm \label{dethamming}
Using the same techniques as in Theorem \ref{derand}, we can sample a probabilistic polynomial from Theorem \ref{thm:poly} with only $O(\log(n)\log(ns))$ random bits.
\end{remark}

\begin{corollary}\label{cor:poly}
Given $d,s,t,\eps$,
we can construct a probabilistic polynomial $\PP:\{0,1\}^{ns}\rightarrow\R$ of degree at most $\D := O((1/\eps)^{1/3}\log s)$ with at most
$s\cdot \binom{n}{D}$ monomials, such that
\begin{compactitem}
\item if $\bigvee_{i=1}^s \predicate{\sum_{j=1}^n x_{ij} \ge t}$ is false, then $|\PP(x_{11},\ldots,x_{1n},\ldots,x_{s1},\ldots,x_{sn})|\le s$ with probability at least $2/3$;
\item if $\bigvee_{i=1}^s \predicate{\sum_{j=1}^d x_{ij} \ge t+\eps n}$ is true, then $\PP(x_{11},\ldots,x_{1n},\ldots,x_{s1},\ldots,x_{sn})> 2s$ with probability at least $2/3$.
\end{compactitem}
For the exact setting with $\eps=1/n$, we can alternatively bound $\D$ by
$O(n^{1/3}\log^{2/3}(ns))$.
\end{corollary}
\begin{proof}
Define
$\PP(x_{11},\ldots,x_{1n},\ldots,x_{s1},\ldots,x_{sn})
\ :=\ \sum_{i=1}^s \PP_{n,3s,t_i,\eps}(x_{i1},\ldots,x_{in}).
$
\end{proof}

\begin{remark}\rm 
The coefficients of the polynomials from Fact~\ref{fact:AlmWil} are $\poly(n)$-bit integers, and it can be checked that the coefficients of all our deterministic and probabilistic polynomials are rational numbers with $\poly(n)$-bit numerators and a common $\poly(n)$-bit denominator, and that the same bound for the number of monomials holds for the construction time, up to $\poly(n)$ factors. That is, computations with these polynomials have low computational overhead relative to $n$.
\end{remark}

\section{Exact and Approximate Offline Nearest Neighbor Search} \label{hnnsection}

We now apply our new probabilistic PTF construction to obtain a faster  algorithm for offline exact nearest/farthest neighbor search in Hamming space:

\begin{reminder}{Theorem~\ref{thm:ham}}
Given $n$ red and $n$ blue points in $\{0,1\}^d$ for $d=c\log n \ll \log^3 n/\log^5\log n$, we can find an (exact) Hamming nearest/farthest blue neighbor for every red point in randomized time
$n^{2-1/O(\sqrt{c}\log^{3/2}c)}$.
\end{reminder}
\begin{proof}
We proceed as in Abboud, Williams, and Yu's algorithm for Boolean orthogonal vectors~\cite{AbboudWY15} or Alman and Williams' algorithm for Hamming closest pair~\cite{JoshRyan}.
For a fixed $t$,
we first solve the decision problem of testing whether the nearest neighbor
distance is less than $t$ for each red point.  (Farthest neighbors are similar.)
Let $s=n^\alpha$ for some parameter $\alpha$ to be set later.
Arbitrarily divide the blue point set into $n/s$ groups of $s$ points.
For every group $G$ of blue points 
and every red point $q$,
we want to test whether 
\[ F(G,q)\ :=\ \predicate{\min_{p\in G} \|p-q\|_1 < t} \ =\ 
\bigvee_{p\in G} \predicate{\sum_{i=1}^d (p_iq_i + (1-p_i)(1-q_i))
> d-t}
\] 
(where $p_i$ denotes the $i$-th coordinate of a point $p$).  
By Corollary~\ref{cor:poly}, we can express $F(G,q)$ 
as a probabilistic polynomial that has
the following number of monomials:
\begin{eqnarray*}
s \cdot \binom{O(d)}{ O(d^{1/3}\log^{2/3}(ds))}
&\le& n^\alpha\cdot O\left(\frac{c\log n}{c^{1/3}\alpha^{2/3}\log n}\right)^{O(c^{1/3}\alpha^{2/3}\log n)}\\
&\le& n^\alpha\cdot n^{O(c^{1/3}\alpha^{2/3}\log\frac{c}{\alpha})}
\ \ll\ (n/s)^{0.1}
\end{eqnarray*}
for large enough $n$, by setting $\alpha$ to be a sufficiently small constant times $1/(c^{1/3}\log^{3/2}c)$.  The same bound holds for the construction time of the polynomial.

We can rewrite the polynomial for $F(G,q)$ as the dot product of two vectors $\phi(G)$ and $\psi(q)$ in $(n/s)^{0.1}$ dimensions over $\R$.
The problem of evaluating $F(G,q)$ over all $n/s$ groups $G$ of
blue points and all red points $q$ then reduces to
multiplying an $n/s\times (n/s)^{0.1}$ with an $(n/s)^{0.1}\times n$
matrix over $\R$.  This in turn reduces to $s$ instances of multiplication of
$n/s\times (n/s)^{0.1}$ with $(n/s)^{0.1}\times n/s$ matrices,
each of which can be done in $\OO(n/s)^2$ arithmetic operations on $\poly(d)$-bit numbers over an appropriately large field (Lemma~\ref{rectangular}). The total time is $\OO(\poly(d)n^2/s)=O(n^{2-1/O(c^{1/3}\log^{3/2}c)})$.  

The error probability for each pair $(G,q)$ is at most $1/3$, which can be lowered to $O(1/n^3)$, for example,
by repeating $O(\log n)$ times (and taking the majority of the answers).
The overall error probability is then $O(1/n)$.  This solves the
decision problem for a fixed $t$, but we can compute all nearest neighbor
distances by calling the decision algorithm $d$ times for all values of $t$.
For each red point, we can find an actual nearest neighbor in additional $O(s)$ time,
since we know which group achieves the nearest neighbor distance. 
\end{proof} \smallskip

The same approach can be applied to solve \emph{approximate} nearest neighbor search in Hamming space:

\begin{theorem}\label{thm:ham:approx0} 
Given $n$ red and $n$ blue points in $\{0,1\}^d$ and $\eps\gg \log^6(d\log n)/\log^3 n$, we can find 
an approximate Hamming nearest/farthest blue neighbor with additive error at most $\eps d$ for each red point in randomized time
$n^{2-\Omega(\eps^{1/3}/\log(\frac{d}{\eps\log n}))}$.
\end{theorem}
\begin{proof} We mimic the proof of Theorem~\ref{thm:ham} up to the definition of the polynomial $F(G,q)$. However, instead of applying the exact polynomial of Corollary~\ref{cor:poly}, we insert the \emph{approximate} polynomial construction from the same corollary. While the exact polynomial had degree $O(d^{1/3}\log^{2/3}(ds))$, the approximate one has degree $O((1/\epsilon)^{1/3} \log s)$. 
Setting \[s := n^{\alpha} := n^{\Omega(\eps^{1/3}/\log(\frac{d}{\eps\log n}))},\] 
 the number of monomials in the new polynomial is now
\begin{eqnarray*}
s\cdot \binom{O(d)}{O((1/\eps)^{1/3}\log s)}
&\le&
n^\alpha\cdot O\left(\frac{d}{(\alpha/\eps^{1/3})\log n}\right)^{O((\alpha/\eps^{1/3})\log n)}\\
&\le&
n^\alpha\cdot n^{O((\alpha/\eps^{1/3})\log\frac{d}{\alpha\log n})}
\ \ll\ (n/s)^{0.1},
\end{eqnarray*}
for large enough $n$. The remainder of the algorithm is the same as the proof of Theorem~\ref{thm:ham}, and the running time is $\tilde{O}(n^2/s^2) \leq n^{2-\Omega(\eps^{1/3}/\log(\frac{d}{\eps\log n}))}$.
\end{proof}

\begin{remark}\rm\label{rmk:det:alg}
For deterministic algorithms, using Corollary~\ref{cor:det:poly} instead,
the time bounds for Theorems~\ref{thm:ham} and~\ref{thm:ham:approx}
become $n^{2-1/O(c\log^2c)}$ and $n^{2-\Omega(\sqrt{\eps}/\log(\frac{d}{\eps\log n}))}$ respectively.
\end{remark}

The algorithm of Theorem~\ref{thm:ham:approx0} still has three drawbacks: (i)~the exponent in the time bound depends on the dimension~$d$, (ii)~the result requires additive instead of multiplicative error, and (iii)~the result is for Hamming space instead of more generally $\ell_1$ or $\ell_2$.  We can resolve all three issues at once, by using known dimension reduction techniques:

\begin{reminder}{Theorem~\ref{thm:ham:approx}}
Given $n$ red and $n$ blue points in $[U]^d$ and $\eps\gg \frac{\log^6\log n}{\log^3 n}$, we can find 
a $(1+\eps)$-approximate $\ell_1$ or $\ell_2$ nearest/farthest blue neighbor for each red point in 
$(dn+n^{2-\Omega(\eps^{1/3}/\log(1/\eps))})\cdot \poly(\log(nU))$ randomized time.

\end{reminder}
\begin{proof} {\bf (The $\ell_1$ case.)} We first solve the decision problem for a fixed threshold value $t$.
We use a variant of $\ell_1$ locality-sensitive hashing (see~\cite{AndoniMasterThesis}) to map points from $\ell_1$ into low-dimensional Hamming space
(providing an alternative to Kushilevitz, Ostrovsky, and Rabani's dimension reduction technique for Hamming space~\cite{KOR}).
For each red/blue point $p$ and each $i\in\{1,\ldots,k\}$,
define $h_i(p)=(h_{i1}(p),\ldots,h_{id}(p))$ with $h_{ij}(p) = \down{(p_{a_{ij}}+b_{ij})/(2t)}$
where $a_{ij}\in\{1,\ldots,d\}$ and $b_{ij}\in [0,2t)$ are independent uniformly distributed random variables.
For each of the $O(n)$ hashed values of $h_i$, pick a random bit; let $f_i(p)$ be the random bit associated with $h_i(p)$.
Finally, define $f(p)=(f_1(p),\ldots,f_k(p))\in\{0,1\}^k$. For any fixed $p,q$, 
\begin{eqnarray*}
 \Pr[h_{ij}(p)\neq h_{ij}(q)] &=& \frac{1}{d}\sum_{a=1}^d \min\left\{\frac{|p_a-q_a|}{2t}, 1\right\}\\
  \Pr[f_i(p)\neq f_i(q)] &=& \frac{1}{2}\Pr[h_i(p)\neq h_i(q)]\ =\ \frac{1}{2}\Pr\left[\bigvee_{j=1}^k \predicate{h_{ij}(p)\neq h_{ij}(q)}\right].
\end{eqnarray*}
\begin{itemize}
\item If $\|p-q\|_1\le t$, then $\Pr[h_{ij}(p)\neq h_{ij}(q)]\le \frac{\|p-q\|_1}{2dt}\le \frac{1}{2d}$ and 
$\Pr[f_i(p)\neq f_i(q)] \le \alpha_0 := \frac{1}{2}(1 - (1-\frac{1}{2d})^d)$;
\item if $\|p-q\|_1 \ge (1+\eps)t$, then $\Pr[h_{ij}(p)\neq h_{ij}(q)]\ge
\min\{\frac{\|p-q\|_1}{2dt},\frac{1}{d}\}\ge \frac{1+\eps}{2d}$ and 
$\Pr[f_i(p)\neq f_i(q)] \ge \alpha_1 := \frac{1}{2}(1 - (1-\frac{1+\eps}{2d})^d)$.
\end{itemize}
Note that $\alpha_1-\alpha_0=\Omega(\eps)$.
By a Chernoff bound, it follows (assuming $k\ge\log n$) that
\begin{itemize}
\item if $\|p-q\|_1\le t$, then $\|f(p)-f(q)\|_1\le A_0 := \alpha_0 k + 
O(\sqrt{k\log n})$ with probability $1-O(1/n^3)$;
\item if $\|p-q\|_1\ge (1+\eps)t$, then $\|f(p)-f(q)\|_1\ge A_1 :=
\alpha_1 k - O(\sqrt{k\log n})$ with probability $1-O(1/n^3)$.
\end{itemize}
Note that $A_1-A_0=\Omega(\eps k)$ by setting $k$ to be a sufficiently 
large constant times $(1/\eps)^2\log n$.
We have thus reduced the problem to an approximate problem
with additive error $O(\eps k)$ for Hamming space in $k=O((1/\eps^2)\log n)$ dimensions,
which by Theorem~\ref{thm:ham:approx0} requires $n^{2-\Omega(\eps^{1/3}/\log(1/\eps))}$ time.  The initial cost of applying the mapping $f$
is $O(dkn)$.

This solves the decision problem; we can solve the original problem
by calling the decision algorithm $O(\log_{1+\eps}U)$ times for all $t$'s that
are powers of $1+\eps$. 
\end{proof} \smallskip

\begin{proof} {\bf (The $\ell_2$ case.)} We use a version of the Johnson--Lindenstrauss lemma to map from $\ell_2$ to $\ell_1$ (see for example~\cite{Mat08}).  For each red/blue point $p$, define 
$f(p)=(f_1(p),\ldots,f_k(p))\in\R^k$ with $f_i(p)=\sum_{j=1}^k a_{ij}p_j$, where the $a_{ij}$'s are independent normally distributed random variables with mean 0 and variance~1.  For each fixed $p,q\in\R^d$, 
it is known that after rescaling by a constant,
$\|f(p)-f(q)\|_1$ approximates $\|p-q\|_2$ to within 
$1\pm O(\eps)$ factor with probability $1-O(1/n^3)$, by setting $k=O((1/\eps)^2\log n)$.  It suffices to keep $O(\log U)$-bit precision of the mapped points.  The initial cost of applying the
mapping $f$ is $O(dkn)$ (which can be slightly improved by utilizing a sparse
Johnson--Lindenstrauss transform~\cite{AilonChazelle}).
\end{proof} \smallskip

Numerous applications to high-dimensional computational geometry now follow.  We briefly mention just one such application, building on the work of \cite{IndykMotwani,HarPeledIndykMotwani}:

\begin{corollary}\label{cor:mst}
Given $n$ points in $[U]^d$ and $\eps\gg \log^6\log n/\log^3 n$, we can find 
a $(1+\eps)$-approximate $\ell_1$ or $\ell_2$ minimum spanning tree
in 
$(dn+n^{2-\Omega(\eps^{1/3}/\log(1/\eps))})\cdot \poly(\log(nU))$ randomized time.
\end{corollary}

\begin{proof}
Let $G_r$ denote the graph where the vertex set is the given point set $P$
and an edge $pq$ is present whenever $p$ and $q$ have distance at most $r$.
Har-Peled, Indyk, and Motwani~\cite{HarPeledIndykMotwani}
gave a reduction of the approximate minimum spanning tree
problem to the following \emph{approximate connected components} problem:
\begin{quote}
Given a value $r$, compute a partition of $P$ into
subsets with the properties that
(i)~two points in the same subset must be in the same component 
in $G_{(1+\eps)r}$, and
(ii)~two points in different subsets must be in different components
in $G_r$.
\end{quote}
The reduction is based on Kruskal's algorithm and
increases the running time by a logarithmic factor.

To solve the approximate connected components problem,
Har-Peled, Indyk, and Motwani gave a further reduction to
online dynamic approximate nearest neighbor search.  Since we want
a reduction to offline static approximate nearest neighbor search, we proceed differently.

We first reduce the approximate connected components problem to 
the \emph{offline approximate nearest foreign neighbors} problem:
\newcommand{\NFN}{\textrm{NFN}}
\begin{quote}
Given a set $P$ of $n$ colored points with colors from $[n]$, 
for each point $q\in P$, find a $(1+\eps)$-approximate nearest neighbor $\NFN_q$ 
among all points in $P$ with color different from $q$'s color.
\end{quote}
The reduction can be viewed as a variant of Boruvka's algorithm and
is as follows:  Initially assign each point a unique
color and mark all colors as active. 
At each iteration, solve the offline approximate nearest
foreign neighbors problem for points with active colors.  
For each $q$, if $\NFN_q$ and $q$ have
distance at most $(1+\eps)r$ and have different colors, 
merge the color class of $\NFN_q$ and $q$.  If a color class has not been
merged to other color classes during the iteration, 
mark its color as inactive.
When all colors are inactive, output the color classes.
Otherwise, proceed to the next iteration.  The correctness of
the algorithm is obvious.  Since each iteration decreases the number of
active colors by at least a half, the number
of iterations is bounded by $O(\log n)$.  Thus, the reduction
increases the running time by a logarithmic factor.

To finish, we reduce the offline approximate nearest foreign neighbors problem
to the standard (red/blue) offline approximate nearest neighbors problem
by a standard trick:
For each $j=1,\ldots,\lceil\log n\rceil$, for each point $q\in P$ where the
$j$-th bit of $q$'s color is 0 (resp.\ 1), compute an approximate nearest neighbor
of $q$ among all points $p\in P$ where the $j$-th bit of $p$'s color
is 1 (resp.\ 0).  Record the nearest among all approximate nearest
neighbors found for each point $q$.  The final reduction increases
the running time by another logarithmic factor.
\end{proof}

\section{Faster Algorithms For MAX-SAT}
\label{appendix-MAX-SAT}

Next, we apply our improved probabilistic PTFs to obtain faster algorithms for MAX-SAT for sparse instances with $cn$ clauses.  We first consider MAX-$k$-SAT for small $k$ before solving the general problem:

\begin{theorem}\label{thm:maxksat}
Given a $k$-CNF formula $F$ (or $k$-CSP instance) with $n$ variables and $cn\ll n^4/(k^4\log^6 n)$ clauses, we can find an assignment that satisfies the maximum number of clauses (constraints) of $F$ in randomized $2^{n-n/O(k^{4/3}c^{1/3}\log(kc))}$ time.
\end{theorem}

\begin{proof}
We proceed as in the \#$k$-SAT algorithm of Chan and Williams~\cite{ChaWil}.
We first solve the decision problem of testing whether there is a variable
assignment satisfying more than $t$ clauses for a fixed~$t \in [cn]$.
Let $s=\alpha n$ for some parameter $\alpha<1/2$ to be set later.

For $j \in [c_n]$, define the function $C_j(x_1,\ldots,x_n)=1$ if the $j$-th clause of the given formula is satisfied, and $0$ otherwise. Note that  
each $C_j$ can be expressed as a polynomial of degree at most $k$.  

Say that a variable is \emph{good} if it occurs in at most $2kc$ clauses.
By the pigeonhole principle, at least half of the variables are good, so
we can find $s$ good variables $x_1,\ldots,x_s$.  
Let $x_{s+1},\ldots,x_n$ be the remaining variables, and let $J \subset [cn]$ be the set of indices of all clauses $C_j$ that contain some occurrence of a good variable; note that $|J|=O(kcs)$. Now for every variable assignment $(x_{s+1},\ldots,x_n)\in\{0,1\}^{n-s}$,
we want to compute
\[
 F(x_{s+1},\ldots,x_n)\ :=\ 
 \bigvee_{(a_1,\ldots,a_s)\in\{0,1\}^s}
 \predicate{\sum_{j=1}^{cn} C_j(a_1,\ldots,a_s,x_{s+1},\ldots,x_n) > t}.
 \]
We will achieve this by computing for every $t'\in [cn]$:
\[
 G_{t'}(x_{s+1},\ldots,x_n)\ :=\ 
 \bigvee_{(a_1,\ldots,a_s)\in\{0,1\}^s}
 \predicate{\sum_{j\in J} C_j(a_1,\ldots,a_s,x_{s+1},\ldots,x_n) > t'}.
 \]

Let us define $T[x_{s+1},\ldots,x_n] := t-\sum_{j\not\in J}C_j(0,\ldots,0,x_{s+1},\ldots,x_n)$. (Observe that it is OK to zero out the good variables $x_1,\ldots,x_s$ here, because we are only summing over clauses that \emph{do not} contain them.) Note that $T$ can be viewed as a polynomial in $n-s$ variables with only $\poly(n)$ monomials. Therefore for all $(x_{s+1},\ldots,x_n) \in \{0,1\}^{n-s}$, these $T$-values can be precomputed in $\poly(n)2^{n-s}$ time. As these $T$-values are measuring the contribution from the variables $x_{s+1},\ldots,x_n$ to the number of satisfied clauses, we have \[F(x_{s+1},\ldots,x_n)= G_{T[x_{s+1},\ldots,x_n]}(x_{s+1},\ldots,x_n).\]
Applying Corollary~\ref{cor:poly} (in the exact setting), we can express any $G_{t'}$ as a sum of $2^s$ probabilistic
polynomials of degree $k \cdot O((kcs)^{1/3}(s+\log(kcs))^{2/3})$, where each probabilistic polynomial computes an expression of the form $\predicate{\sum_{j \in J} p_j(x_{s+1},\ldots,x_n)}$ with error probability at most $1/(10\cdot 2^s)$, and for all $j \in J$ we have $\deg(p_j(x_{s+1},\ldots,x_n)) \leq k$. The number of monomials in our probabilistic polynomial for $G_{t'}$ is at most
\begin{eqnarray*}
 2^s\cdot \binom{n-s}{k\cdot O((kcs)^{1/3}(s+\log(kcs))^{2/3})} 
 &\le& 
2^{\alpha n} \cdot O\left(\frac{n}{k^{4/3}c^{1/3}\alpha n}\right)^{O(k^{4/3}c^{1/3}\alpha n)}\\
 &\le & 2^{\alpha n}\cdot 2^{O(k^{4/3}c^{1/3}\alpha\log\frac{1}{\alpha})n}\ \ll\ 2^{0.1n}
\end{eqnarray*}
by setting $\alpha$ to be a sufficiently small constant times
$1/(k^{4/3}c^{1/3}\log(kc))$.  The same bound holds for the
construction time of the polynomial.

For each $t'$,
we can evaluate the polynomial for $G_{t'}$ at all $2^{n-s}$ input values by divide-and-conquer or dynamic programming using
$\poly(n)2^{n-s}$ arithmetic operations~\cite{Yates,WilliamsJACM14} on $\poly(n)$-bit numbers.
The total time is $2^{n-n/O(k^{4/3}c^{1/3}\log(kc))}$.  As before,
the error probability can be lowered by taking the majority values over $O(n)$ repetitions, and the original problem can be solved
by calling the decision algorithm for at most $cn$ times.
\end{proof}

\begin{reminder}{Theorem~\ref{thm:maxsat}}
Given a CNF formula with $n$ variables and $cn\ll n^4/\log^{10} n$ clauses,
we can find an assignment that satisfies the maximum number of clauses
in randomized $2^{n-n/O(c^{1/3}\log^{7/3}c)}$ time.
\end{reminder}
\begin{proof}
We use a standard width reduction technique~\cite{SakSetTam} originally observed by Schuler~\cite{Schuler05} and studied closely by Calabro, Impagliazzo, and Paturi~\cite{CIP06}. Consider the following recursive algorithm:
\begin{compactitem}
\item If all clauses have length at most $k$, then call the algorithm from Theorem~\ref{thm:maxksat} and return its output.
\item Otherwise, pick a clause $(\alpha_1\vee\cdots\vee\alpha_\ell)$ with 
$\ell>k$. Return ``SAT'' if at least one of the two following calls return ``SAT'':
\begin{compactitem}
\item  Recursively solve the instance in which $(\alpha_1\vee\cdots\vee\alpha_\ell)$ is replaced by $(\alpha_1\vee\cdots\vee\alpha_k)$, and  
\item recursively solve the instance in which $\alpha_1,\ldots,\alpha_k$ are all assigned \emph{false}.
\end{compactitem}
\end{compactitem}
Sakai, Seto, and Tamaki's analysis for MAX-SAT~\cite{SakSetTam} can be directly
modified to show that the total time of this algorithm remains
$2^{n-n/O(k^{4/3}c^{1/3}\log(kc))}$, when the parameter $k$ is set to be a sufficiently large constant times $\log c$.
\end{proof}

For MAX-$k$-SAT with $k\le 4$, we can obtain a much better dependency on the sparsity parameter~$c$; in fact, we obtain significant speedup even for general dense instances.  The approach this time requires only the previous probabilistic polynomials by Alman and Williams~\cite{JoshRyan}.  Naively, the dense case seems to require threshold functions with superlinearly many arguments, but by incorporating a few new ideas, we manage to solve MAX-4-SAT using only $O(n)$-variate threshold functions.

\begin{reminder}{Theorem~\ref{thm:max4sat}}
Given a weighted 4-CNF formula $F$ with $n$ variables with positive integer weights bounded by $\poly(n)$, we can find an assignment that maximizes the total weight of clauses satisfied in $F$, in randomized $2^{n-n/O(\log^2 n\log^2\log n)}$ time.
In the sparse case when the clauses have total weight $cn$, the time bound improves
to $2^{n-n/O(\log^2 c\log^2\log c)}$.
\end{reminder}
\begin{proof} {\bf (Dense case.)}\ \ 
Let $s=\alpha n$ for some parameter $\alpha$ to be set later. Arbitrarily divide the $n$ variables of $F$ into three groups: 
$x=\{x_1,\ldots,x_{(n-s)/2}\}$,
$y=\{y_1,\ldots,y_{(n-s)/2}\}$,
and $z=\{z_1,\ldots,z_s\}$.
As in Theorem~\ref{thm:maxksat}, it suffices to solve the decision problem of whether there exist $x,y\in\{0,1\}^{(n-s)/2}$ and $z\in\{0,1\}^s$ such that $f(x,y,z)>t$, for a given degree-4 polynomial $f$ and a fixed $t \in [n^{c_0}]$ (for an appropriately large constant $c_0$).
 Since $f$ has degree 4, observe that each term has either (a) at most one $y$ variable, (b) at most one $x$ variable, or (c) no $z$ variable.
We can thus write 
\[ f(x,y,z) = \sum_{i=1}^{(n-s)/2}\! f_i(x,z) y_i  + \sum_{i=1}^{(n-s)/2}\! g_i(y,z) x_i
+ h(x,y)
\]
where the $f_i$'s and $g_i$'s are degree-3 polynomials, and $h$ is a degree-4 polynomial.

For every $x,y\in\{0,1\}^{(n-s)/2}$, it suffices to compute
\[ F(x,y) := \sum_{z\in\{0,1\}^s} \predicate{f(x,y,z)>t}.
\]
More generally, we compute for every $t'\in [n^{c_0}]$:
\[ G_{t'}(x,y) := \sum_{z\in\{0,1\}^s} H_{z,t'}(x,y),\ \ \mbox{with}\ \
H_{z,t'}(x,y) := \predicate{ \sum_{i=1}^{(n-s)/2}\! f_i(x,z) y_i  + \sum_{i=1}^{(n-s)/2}\! g_i(y,z) x_i > t' }.
\]
Then $F(x,y) = G_{t-h(x,y)}(x,y)$; we can precompute all $h(x,y)$
values in $\poly(n)2^{n-s}$ time.  

The $H_{z,t'}(x,y)$ predicate can be viewed as a \emph{weighted} threshold function with $O(n)$ arguments.  To further complicate matters, these weights are not fixed: they depend on $x$ and $y$.  We resolve the issue by extending the vectors $x$ and $y$ and using a binary representation trick.

For each vector $x\in\{0,1\}^{(n-s)/2}$, define 
an \emph{extended vector} $x^*$ where
$x^*_i=x_i$ for each $i=1,\ldots,(n-s)/2$ and
$x^*_{i,j,z}$ is the $j$-th least significant bit in the binary representation of $f_i(x,z)$ for each $i=1,\ldots,(n-s)/2$,
$j=0,\ldots,\ell$ and $z\in\{0,1\}^s$, with $\ell=O(\log n)$.
Note that $x^*$ is a vector in $O(n\cdot \log n \cdot 2^s)$ dimensions.
Similarly, for each vector $y\in\{0,1\}^{(n-s)/2}$, define
an extended vector $y^*$ where
$y^*_i=y_i$ for each $i=1,\ldots,(n-s)/2$ and
$y^*_{i,j,z}$ is the $j$-th least significant bit in the binary representation of $g_i(y,z)$ for each $i=1,\ldots,(n-s)/2$,
$j=0,\ldots,\ell$ and $z\in\{0,1\}^s$. We can precompute all extended vectors in $2^{(n-s)/2}\cdot \poly(n)2^s$ time.

Then
\[ H_{z,t'}(x,y) := \sum_{(t_0,\ldots,t_\ell)}\prod_{j=0}^{\ell}\predicate{ \sum_{i=1}^{(n-s)/2}\! x^*_{i,j,z} y_i  + \sum_{i=1}^{(n-s)/2}\! y^*_{i,j,z} x_i = t_j },
\]
where the outer sum is over all tuples $(t_0,\ldots,t_\ell)\in [n^{c_0}]^\ell$ with
$\sum_{j=0}^{\ell} 2^j\cdot t_j > t'$.

By Fact~\ref{fact:AlmWil}, for each $z \in \{0,1\}^s$, $j=0,\ldots,\ell$, and $t_j \in [n^{c_0}]$, we can construct a probabilistic polynomial 
(over $\R$ or $\mathbb{F}_2$)
for the predicate $\predicate{ \sum_i x^*_{i,j,z} y_i  + \sum_i y^*_{i,j,z} x_i = t_j }$
with degree $O(\sqrt{n\log S})$ with error probability at most $1/S$.
By the union bound, the probability that there is an error for some $z,j,t_j$ is at most 
$O((1/S)\cdot 2^s\cdot\log n\cdot n^{O(1)})$, which can be made at most $1/4^s$, for example,
by setting $S=n^{c_0}2^s$ for a sufficiently large constant $c_0$.
Thus, the degree for each predicate is $O(\sqrt{ns})$ (assuming $s\ge\log n$).

For each $z \in \{0,1\}^s$ and $t' \in [n^{c_0}]$, by distributing over the product $\prod_{j=0}^{\ell}$ we can then construct a probabilistic polynomial for $H_{z,t'}(x,y)$ with degree $O(\sqrt{ns}\ell) \leq O(\sqrt{ns}\log n)$. For a fixed $z$ and $t'$, such a polynomial is a function of $O(n\log n)$ free variables in $x^*$ and $y^*$, and therefore has at most
$\binom{O(n\log n)}{O(\sqrt{ns}\log n)}$ monomials.  The same bound holds for the time needed to construct the probabilistic polynomial (note the number of tuples $(t_0,\ldots,t_\ell)$ is $n^{O(\log n)}$, which is a negligible factor).

For each $t' \in [n^{c_0}]$, we can thus construct a probabilistic polynomial for $G_{t'}(x,y)$ with degree $O(\sqrt{ns}\log n)$ over $x^*$ and $y^*$, with the following number of monomials:
\begin{eqnarray*}
 2^s\cdot \binom{O(n\log n)}{O(\sqrt{ns}\log n)} &\le & 
   2^{\alpha n}\cdot O\left(\frac{n\log n}{\sqrt{\alpha} n\log n}\right)^{O(\sqrt{\alpha}n\log n)}\\
   &\le & 2^{\alpha n}\cdot 2^{\sqrt{\alpha}n(\log (n))\log(1/\alpha)}
   \ \ll\ 2^{0.1(n-s)/2}
\end{eqnarray*}
by setting $\alpha$ to be a sufficiently small constant times
$1/(\log n \cdot \log\log n)^2$.  The same bound holds for the construction time.

We can rewrite the polynomial for $G_{t'}(x,y)$ as the dot product of two
vectors $\phi(x^*)$ and $\psi(y^*)$ of $2^{0.1(n-s)/2}$ dimensions.
The problem of evaluating $G_{t'}(x,y)$ over all $x,y\in\{0,1\}^{(n-s)/2}$ then reduces to multiplying a $2^{(n-s)/2}\times 2^{0.1(n-s)/2}$ with
a $2^{0.1(n-s)/2}\times 2^{(n-s)/2}$ matrix (over $\R$ or $\mathbb{F}_2$), which can be done in $\poly(n)2^{n-s}$ time (Lemma~\ref{rectangular}).  The total time is
$2^{n-n/O(\log^2 n\log^2\log n)}$.
\end{proof}

\begin{proof} {\bf (Sparse case.)}
If the clauses have total weight $cn$,
we can refine the analysis above, in the following way.
Let $\mu_i$ and $\nu_i$ be the maximum value of $f_i(x,z)$ and $g_i(y,z)$
respectively.  We know that $\sum_i (\mu_i+\nu_i) \le cn$.
The variable $x^*_{i,j,z}$ is needed only
when $j\le \log(\mu_i)$, and the variable $y^*_{i,j,z}$ is
needed only when $j\le\log(\nu_i)$.
For each $z,j,t_j$, the probabilistic polynomial
for the predicate \[\predicate{ \sum_i x^*_{i,j,z} y_i  + \sum_i y^*_{i,j,z} x_i = t_j }\] has degree $O(\sqrt{n_js})$, where $n_j$ is the number of $i$'s
with $\mu_i\ge 2^j$ or $\nu_i\ge 2^j$.

Observe that $n_j=O(cn/2^j)$.
It follows that the degree for the $H_{z,t'}(x,y)$ polynomial
is $O(\sum_{j=0}^\ell \sqrt{n_js})
=O(\sqrt{ns}\log c + \sum_{j>\log c} \sqrt{(cn/2^j)s})
=O(\sqrt{ns}\log c)$.
The number of variables in $H_{z,t'}(x,y)$ is
at most $O(\sum_{j=0}^\ell n_j) = O(n\log c + \sum_{j>\log c} (cn/2^j)
)=O(n\log c)$.

Thus, the bound on the total number of monomials becomes
\begin{eqnarray*}
 2^s\cdot \binom{O(n\log c)}{O(\sqrt{ns}\log c)} &\le & 
   2^{\alpha n}\cdot O\left(\frac{n\log c}{\sqrt{\alpha} n\log c}\right)^{O(\sqrt{\alpha}n\log c)}\\
   &\le & 2^{\alpha n}\cdot 2^{\sqrt{\alpha}n\log c\log(1/\alpha)}
   \ \ll\ 2^{0.1(n-s)/2}
\end{eqnarray*}
by setting $\alpha$ to be a sufficiently small constant times
$1/(\log c\log\log c)^2$.
\end{proof}

\section{Circuit Satisfiability Algorithms}
\label{appendix-Circuit-SAT}

In this section, we give new algorithms for solving the SAT problem on some rather expressive circuit classes. First, we outline some notions used in both algorithms.

\subsection{Satisfiability on a Cartesian Product}

In intermediate stages of our SAT algorithms, we will study the following generalization of SAT, where the task is to find a SAT assignment in a ``Cartesian product'' of possible assignments.

\begin{definition} Let $n$ be even, and let $A, B \subseteq \{0,1\}^{n/2}$ be arbitrary. The \emph{SAT problem on the set $A \times B$} is to determine if a given $n$-input circuit has a satisfying assignment contained in the set $A \times B$.
\end{definition}

Recall that a Boolean function $f: \{0,1\}^n \rightarrow \{0,1\}$ is a linear threshold function (LTF) if there are $a_1,\ldots,a_n, t \in \R$ such that for all $x \in \{0,1\}^n$, $f(x) = 1 \iff \sum_i a_i x_i \geq t$.

Let ${\sc Circuit \circ \LTF}[Z,S]$ be the class of circuits with a layer of $S$ LTFs at the bottom (nearest the inputs), with $Z$ additional arbitrary gates above that layer. Let ${\sc Circuit \circ \SUM \circ \AND}[Z,S]$ be the analogous circuit class, but with $S$ DNFs at the bottom layer with property that each DNF always has at most \emph{one} conjunct true for every variable assignment. (Thus we may think of the DNF as simply an \emph{integer sum}.) We first prove that the SAT problem for ${\sc Circuit\circ \LTF}$ can be reduced to the SAT problem for ${\sc Circuit\circ \SUM \circ \AND}$, utilizing a weight reduction trick that can be traced back to Matou\v sek's algorithm for computing dominances in high dimensions~\cite{Matousek91,WilliamsACCTHR14}:

\begin{lemma} \label{LTF-to-AND} Let $A, B \subseteq \{0,1\}^{n/2}$, with $|A|=|B|=N \leq 2^n$. Let $K \in [1,N]$ be an integer parameter. The SAT problem for ${\sc Circuit\circ \LTF}[Z,S]$ circuits on the set $A \times B$ can be reduced to the SAT problem for ${\sc Circuit\circ \SUM \circ \AND}[Z, S]$ where each DNF has at most $O(\log K)$ terms and each $\AND$ has fan-in at most $2\log K$, on a prescribed set $A' \times B'$ with $|A'|=|B'|=N$ and $A',B' \subseteq \{0,1\}^{2S\log K}$. The reduction has the property that if the latter SAT problem can be solved in time $T$, then the former SAT problem can be solved in time $\left(T + N^2 \cdot Z^2/K + N\cdot S\right)\cdot \poly(n)$.
\end{lemma}

\begin{proof} For a given circuit $C$ of type ${\sc Circuit\circ \LTF}[Z,S]$, let the $j$th LTF in the bottom layer have weights $\alpha_{j,1}, \ldots, \alpha_{j,n}, t_j$. Let the assignments in $A$ be $a_1,\ldots,a_N$, and let the assignments in $B$ be $b_1,\ldots,b_N$. Denote the $k$th bit of $a_i$ and $b_i$ as $a_i[k]$ and $b_i[k]$, respectively. 

Make $N \times S$ matrices $M_A$ and $M_B$, where \[M_A[i,j] = \sum_{k=1}^{n/2} \alpha_{j,k}\cdot a_i[k]\] and \[M_B[i,j] = t_j - \sum_{k=1}^{n/2} \alpha_{j,n/2+k}\cdot b_i[k].\] The key property of these matrices is that $M_A[i,j] \geq M_B[i',j]$ if and only if the $n$-variable assignment $(a_i,b_{i'})$ makes the $j$th LTF output $1$.

For each $j =1,\ldots,S$, let $L_j$ be the list of all $2\cdot N$ entries in the $j$th column of $M_A$ and the $j$th column of $M_B$, sorted in increasing order. Partition $L_j$ into $K$ contiguous parts of $O(N/K)$ entries each, and think of each part of $L_j$ as containing a set of $O(N/K)$ assignments from $A \cup B$. (So, the partition of $L_j$ is construed as a partition of the assignments in $A \cup B$.) 
There are two possible cases for a satisfying assignment to the circuit $C$:
\begin{enumerate}
\item \emph{There is a satisfying assignment $(a_i,b_{i'}) \in A \times B$ such that for some $j = 1,\ldots,S$, $a_i$ and $b_{i'}$ are in the same part of $L_j$.} By enumerating every $a_i \in A$, every $j=1,\ldots,S$, and all $O(N/K)$ assignments $b_{i'}$ of $B$ which are in the same part of $L_j$ as $a_i$, then evaluating the circuit $C$ on the assignment $(a_i,b_{i'})$ in $Z^2\cdot \poly(n)$ time, we can determine satisfiability for this case in $O(N \cdot N/K \cdot Z^2)\cdot \poly(n)$ time. If this does not uncover a SAT assignment, we move to the second case.
\item \emph{There is a satisfying assignment $(a_i,b_{i'}) \in A \times B$ such that for every $j = 1,\ldots,S$, $a_i$ and $b_{i'}$ are different parts of $L_j$.} Then for every LTF gate $j=1,\ldots,S$ on the bottom layer of the circuit, we claim that the $j$-th LTF can be replaced by a sum of $O(\log K)$ $\AND$s on $2\log K$ new variables. In particular, for the $j$-th LTF we define one new set of $\log K$ variables which encodes the index $k=1,\ldots,K$ such that $a_i$ is in part $k$ of $L_j$, and another set of $\log K$ variables which encodes the index $k'$ such that $b_{i'}$ is in part $k'$ of $L_{j}$. Then, determining $\predicate{k \geq k'}$ is equivalent to determining whether $(a_i,b_{i'})$ satisfies the $j$-th LTF gate. Finally, note that the predicate $\predicate{k \geq k'}$ can be computed by a DNF of $O(\log K)$ conjuncts. (Take an OR over all $\ell=0,\ldots,\log K$, guessing that the $\ell$-th bit is the most significant bit in which $k$ and $k'$ differ; we can verify that guess with a conjunction on $2\log K$ variables.) On every possible input $(k,k') \in \{0,1\}^{2 \log K}$, the DNF has at most \emph{one} true conjunction. Thus we can construe the OR as simply an \emph{integer sum} of ANDs, as desired. Preparing these new assignments for this new SAT problem takes time $O(N \cdot S)\cdot \poly(n)$.\end{enumerate}
\vspace{-4ex}
\end{proof}

\subsection{Simulating LTFs with AC0 of MAJORITY}

In our SAT algorithms, we will need a way to simulate LTFs with bounded-depth circuits with MAJORITY gates. This was also used in Williams' work on solving ACC-LTF SAT~\cite{WilliamsACCTHR14}, as a black box. However, here we must pay careful attention to the details of the construction. In fact, we will actually have to modify the construction slightly in order for our circuit conversion to work out. Let us review the construction here, and emphasize the parts that need modification for this paper. Recall that $\MAJ$ denotes the majority function. 

\begin{theorem}[Follows from \cite{Maciel-Therien98}, Theorem 3.3]\label{LTF-as-AC0MAJ} Every LTF can be computed by polynomial-size $\AC^0 \circ \MAJ$ circuits. Furthermore, the circuits can be constructed in polynomial time given the weights of the LTF, and the fan-in of each $\MAJ$ gate can be made $n^{1+\eps}$, for every desired $\eps > 0$, and the circuit has depth $O(\log(1/\eps))$.
\end{theorem} 

It will be crucial for our final results that the fan-in of the $\MAJ$ gates can be made arbitrarily close to linear.

\begin{proof} We begin by revisiting the circuit construction of Maciel and Th\'erien~\cite{Maciel-Therien98}, which shows that the addition of $n$ distinct $n$-bit numbers can be performed with polynomial-size $\AC^0 \circ \MAJ$ circuits. The original construction of Maciel and Th\'erien yields $\MAJ$ gates of fan-in $\tilde{O}(n^2)$, which is too large for our purposes. We can reduce the fan-in of $\MAJ$ gates to $O(n^{1+\eps})$ by setting the parameters differently in their construction. Let us sketch their construction in its entirety, then describe how to modify it.

Recall that $\SYM$ denotes the class of symmetric functions. First, we show that addition of $n$ $n$-bit numbers can be done in $\AC^0\circ \SYM$. Suppose the $n$-bit numbers to be added are $A_1,\ldots,A_n$, where $A_i = A_{i,n}\cdots A_{i,1}$ for $A_{j,i} \in \{0,1\}$. Maciel and Th\'erien partition each $A_i$ into $m$ blocks of $\ell$ bits, where $m \cdot \ell = n$. They compute the sum $S_k$ of the $n$ $\ell$-bit numbers in each block $k=1,\ldots,m$, i.e. \[S_k = \sum_{i=1}^n \sum_{j=1}^{\ell} A_{i,(k-1)\ell + j} \cdot 2^{j-1},\] and note that the desired sum is \[z = \sum_{k=1}^m S_k\cdot 2^{(k-1)\ell}.\] Each $S_k$ can be represented in $\ell + \log n$ bits. Maciel and Th\'erien set $\ell = \log n$, so that each $S_k$ is represented by $2 \ell$ bits. They then split each $S_k$ into $\ell$-bit numbers $H_k$ and $L_k$ such that \[S_k = H_k \cdot 2^{\ell} + L_k.\] Note that the ``high'' part $H_k$ corresponds to the ``carry bits'' of $S_k$. They then note that if \[y_1 := \sum_{k=1}^m H_k \cdot 2^{k\ell},\ \ y_2 := \sum_{k=1}^m L_k \cdot 2^{(k-1)\ell},\] we have 

\begin{compactitem}
\item[(a)] $z = y_1 + y_2$, and 
\item[(b)] each bit of $y_i$ is a function of exactly one $H_k$ or $L_k$ for some $k$. In turn, each $L_k$, $H_k$ is a sum of $n \cdot \ell$ $A_{i,j}$'s where each $A_{i,j}$ is multiplied by a power of two in $[0,2^{\ell}]$. Therefore, each bit of $y_i$ can be computed by a $\SYM$ gate of fan-in at most $n \cdot \ell \cdot 2^{\ell} \leq n^2$. \end{compactitem}

We have therefore reduced the addition of $n$ $n$-bit numbers to adding the two $O(n)$-bit numbers $y_1$ and $y_2$, with a layer of $\SYM$ gates. Adding two numbers can be easily computed in $\AC^0$  (see for example \cite{Chandra-Fortune-Lipton85}), so the whole circuit is of the form $\AC^0 \circ \SYM$. 

We wish to reduce the fan-in of the $\SYM$ gates to $O(n^{1+\eps})$ for arbitrary $\eps > 0$. To reduce the fan-in further, it suffices to find a construction that lets us reduce $\ell$. Naturally, we can try to set $\ell = \eps \log n$ for arbitrarily small $\eps \in (0,1)$. Without loss of generality, let us assume $1/\eps$ is an integer. Then, each $S_k$ is represented in $\ell + \log n \leq (1+1/\eps)\ell$ bits. Let $t=1+1/\eps$. If we then split each $S_k$ into $t$ $\ell$-bit numbers $T^{t-1}_k,\ldots,T^0_k$, ranging from high-order to low-order bits, we then have \[S_k = T^{t-1}_k \cdot 2^{(t-1)\ell} + \cdots + T^1_k \cdot 2^{\ell} + T^0_k.\] Defining the $t$ numbers \[y_i := \sum_{k=1}^m T^{i}_k \cdot 2^{(k+i-1)\ell},\] the desired sum is $z = \sum_{i=0}^{t-1} y_i$. Just as before, each bit of $y_i$ is a function of exactly one $T^i_k$ for some $k$, which is a sum of $n \cdot \ell$ $A_{i,j}$'s where each $A_{i,j}$ is multiplied by an integer in $[0,2^{\ell}]$. Hence each bit of $y_i$ can be computed by a $\SYM$ gate of fan-in at most $n \cdot \ell \cdot 2^{\ell} \leq \tilde{O}(n^{1+\eps})$. So with one layer of $\SYM$ gates, we have reduced the $n$ number $n$-bit addition problem to the addition of $t$ $O(n)$-bit numbers $y_0,\ldots,y_{t-1}$. But for $t\leq \log n$, addition of $t$ $n$-bit numbers can be computed by $\AC^0$ circuits of $\poly(n)$-size and \emph{fixed} depth independent of $t$ (see e.g. \cite{Vollmer99}, p.14-15). This completes the description of our $\AC^0 \circ \SYM$ circuit.

Observe that each $\SYM$ gate can be easily represented by an $\OR \circ \AND \circ \MAJ$ circuit. In particular, the OR is over all $j \in \{0,1,\ldots,n\}$ such that the $\SYM$ gate outputs $1$ when given $j$ inputs are equal to $1$, and the $\AND \circ \MAJ$ part computes $\sum_j x_j = j$. Again, the fan-in of each $\MAJ$ here is $\tilde{O}(n^{1+\eps})$. 

We now apply the addition circuits to show how every LTF on $n$ variables can be represented by a polynomial-size $\AC^0 \circ \MAJ$ circuit. Suppose our LTF has weights $w_{1},\ldots,w_{n+1}$, computing $\sum_{j=1}^{n} w_{j}x_{j} \geq w_{n+1}$. By standard facts about LTFs, we may assume for all $j$ that $|w_j|\leq 2^{b n\log_2 n}$ for some constant $b > 0$. Set $W = b n\log_2 n$. 

Let $D$ be a $\AC^0 \circ \MAJ$ circuit for adding $n$ $W$-bit numbers as described above, where each $\MAJ$ gate has fan-in $\tilde{O}(n^{1+\eps})$. For all $j=1,\ldots,n$, connect to the $j$th $W$-bit input of $D$ a circuit which, given $x_{j}$, feeds $w_{j}$ to $D$ if the input bit $x_{i_j}=1$, and the all-zero $W$-bit string if $x_{j}=0$. Observe this extra circuitry is only wires, no gates: we simply place a wire from $x_{j}$ to all bits of the $j$th $W$-bit input where the corresponding bit of $w_{j}$ equals $1$.

This new circuit $D'$ clearly computes the linear form $\sum_{j=1}^n w_{j}x_{j}$. The linear form can then be compared to $w_{n+1}$ with an $\AC^0$ circuit, since the ``less-than-or-equal-to'' comparison of two integers can be performed in $\AC^0$. Indeed, this function can be represented as a quadratic-size DNF ($\SUM \circ \AND$), as was noticed in Lemma~\ref{LTF-to-AND}. We now have an $\AC^0 \circ \MAJ$ circuit $D''$ of size $\poly(W,t) \leq n^b$ computing the LTF, where the $\MAJ$ gates have fan-in $\tilde{O}(n^{1+\eps})$. 
\end{proof}

\subsection{Satisfiability Algorithm for ACC of LTF of LTF}

Let $\AC^0[d,m]\circ \LTF\circ \LTF[S_1,S_2,S_3]$ be the class of circuits with a layer of $S_3$ LTFs at the bottom layer (nearest the inputs), a layer of $S_2$ LTFs above the bottom layer, and a size $S_1$ $\AC^0[m]$ circuit of depth $d$ above the two LTF layers.

\begin{reminder}{Theorem~\ref{ACC-LTF-SAT}} For every integer $d > 0$, $m > 1$, and $\delta > 0$, there is an $\eps > 0$ and an algorithm for satisfiability of $\AC^0[d,m]\circ \LTF\circ \LTF[2^{n^{\eps}},2^{n^{\eps}},n^{2-\delta}]$ circuits that runs in deterministic $2^{n-n^{\eps}}$ time.
\end{reminder}

We use the following depth-reduction theorem of Beigel and Tarui (with important constructibility issues clarified by Allender and Gore~\cite{Allender-Gore94}, and recent size improvements by Chen and Papakonstantinou~\cite{ChenP16}):

\begin{theorem}[\cite{Beigel-Tarui,Allender-Gore94}]\label{BT} Every $\SYM\circ \ACC$ circuit of size $s$ can be simulated by a $\SYM \circ \AND$ circuit of $2^{(\log s)^{c'}}$ size for some constant $c'$ depending only on the depth $d$ and MOD$m$ gates of the $\ACC$ part. Moreover, the $\AND$ gates of the final circuit have only $(\log s)^{c'}$ fan-in, the final circuit can be constructed from the original in $2^{O((\log s)^{c'})}$ time, and the final symmetric function at the output can be computed in $2^{O((\log s)^{c'})}$ time.
\end{theorem}

\begin{proofof}{Theorem~\ref{ACC-LTF-SAT}} Let $\eps > 0$ be a parameter to be set later. The plan is to start with a circuit as specified in the theorem statement, and slowly convert into a nice form that can be evaluated efficiently on many inputs. 

\textbf{1. Trade Variables for Circuit Size.} Our first step is standard for $\ACC$-SAT algorithms~\cite{WilliamsACCTHR14,WilliamsJACM14}: given an $\AC^0[d,m]\circ \LTF\circ \LTF[2^{n^{\eps}},2^{n^{\eps}},n^{2-\delta}]$ circuit $C$ with $n$ variables, create a copy of the circuit $C_{v} := C(v,\cdot)$ for all possible assignments $v \in \{0,1\}^{n^{\eps}}$ to the first $n^{\eps}$ variables of $C$, and define \[C'(x_{n^{\eps}+1},\ldots,x_n) := \bigvee_v C_v(x_{n^{\eps}+1},\ldots,x_n).\] Observe that $C'$ is satisfiable if and only if $C$ is satisfiable, $C'$ has size at most $2^{O(n^{\eps})}$, $C'$ is also an $\AC^0 \circ \LTF \circ \LTF$ circuit, and $C'$ has only $n-n^{\eps}$ variables.

\textbf{2. Replace the middle LTFs with MAJORITYs (Theorem~\ref{LTF-as-AC0MAJ}).} Note that each LTF on the second layer of $C'$ has fan-in at most $n^{2-\delta}+n$, since the number of LTFs on the first layer is $n^{2-\delta}$. Applying the low fan-in transformation of Theorem~\ref{LTF-as-AC0MAJ}, we can replace each of the LTFs on the second layer of $C'$ with $\poly(n)$-size $\AC^0 \circ \MAJ$ circuits where each $\MAJ$ has fan-in at most $n^{2-\delta/2}$. This generates at most $2^{dn^{\eps}}$ new $\MAJ$ gates in the circuit $C'$, for some constant $d > 0$, and produces a circuit of type \[\ACC^0 \circ \MAJ \circ \LTF.\]

\textbf{3. Replace those MAJORITYs with (derandomized) probabilistic polynomials over $\F_2$ (Theorem~\ref{derand}).} We replace each of these new $\MAJ$ gates with our low-randomness probabilistic polynomials for the MAJORITY function, as follows. Recall from Theorem~\ref{derand} that we can construct a probabilistic polynomial over $\F_2$ for $k$-bit MAJORITY with degree $O(\sqrt{k \log(1/\eps')})$ and error at most $\eps'$, using a distribution of $k^{O(\log(k/\eps'))}$ uniformly chosen $\F_2$-polynomials. Setting $k := n^{2-\delta/2}$ for the fan-in of the $\MAJ$ gates, and the error to be $\eps' := 1/2^{2d n^{\eps}}$, the degree becomes \[D := O\left(\sqrt{n^{2-\delta/2} \cdot 2dn^{\eps}}\right) \leq O(n^{1-\delta/4+\eps/2})\] and the sample space has size $S = n^{O(n^{\eps})}$. For $\eps \ll \delta/4$, we have $D := O(n^{1-\delta/8})$, and each polynomial in our sample space has at most $\binom{n^{2-\delta}}{n^{1-\delta/8}} \leq 2^{O(n^{1-\delta/8} \log n)}$ monomials. 
For every choice of the random seed $r$ to the probabilistic polynomial, let $C'_r$ be the circuit $C'$ with the corresponding $\F_2$ polynomial $P_r$ substituted in place of each $\MAJ$ gate. That is, each $\MAJ$ gate is substituted by an $\XOR$ of $2^{O(n^{1-\delta/8} \log n)}$ $\AND$s of fan-in at most $O(n^{1-\delta/8})$.

We now form a circuit $C''$ which takes a majority vote over all $2^{O(n^{\eps} \log n)}$ circuits $C'_r$. The new circuit $C''$ therefore has the form \[\MAJ \circ \ACC^0 \circ \XOR \circ \AND \circ \LTF,\] where the $\MAJ \circ \ACC^0$ part has size $2^{O(n^{\eps} \log n)}$, and each $\XOR \circ \AND \circ \LTF$ subcircuit has size $2^{O(n^{1-\delta/8} \log n)}$. Since our probabilistic polynomial computes MAJORITY with $1/2^{2d n^{\eps}}$ error and there are at most $2^{d n^{\eps}}$ $\MAJ$ gates in $C'$, the new circuit $C''$ is equivalent to the original circuit $C'$. 

\textbf{4. Apply Beigel--Tarui to the top of the circuit, and distribute.} It is very important to observe that we \emph{cannot} apply Beigel--Tarui (Theorem~\ref{BT}) to the \emph{entire} circuit $C''$, as its total size is $2^{\Omega(n^{1-\delta/8} \log n)}$, and the quasi-polynomial blowup of Beigel--Tarui would generate a huge circuit of size $\Omega(2^n)$, rendering our conversion intractable. 

However, the top $\MAJ \circ \ACC^0$ part is still small. Invoking the depth reduction lemma of Beigel and Tarui (Theorem~\ref{BT} above), we can replace the $\MAJ \circ \ACC^0$ part in $C''$ of size $2^{O(n^{\eps} \log n)}$ (even though it has $2^{O(n^{\eps} \log n)}$ inputs from the $\XOR$ layer!) with a $\SYM \circ \AND$ circuit of size $2^{n^{a\cdot\eps}}$ for a constant $a \geq 1$, where each $\AND$ has fan-in at most $n^{a\eps}$, and $a$ depends only on the (constant) depth $d$ and (constant) modulus $m$ of the $\ACC^0$ subcircuit. 

The resulting circuit $C_3$ now has the form \[\SYM \circ \AND \circ \XOR \circ \AND \circ \LTF.\]  Applying the distributive law to the $\AND \circ \XOR$ parts, where the $\AND$s have fan-in at most $n^{a\eps}$ and the $\XOR$s have fan-in $2^{O(n^{1-\delta/8} \log n)}$, each $\AND \circ \XOR$ parts can be converted into an $\XOR \circ \AND$ circuit of size $2^{O(n^{1-\delta/8+a\eps} \log n)}$, where the fan-in of $\AND$s is at most $n^{a\eps}$. Letting $\eps \ll \delta/(ca)$ for sufficiently large $c \geq 1$, the fan-in of the new $\XOR$s is at most $2^{O(n^{1-\eps})}$. We now have a circuit $C_4$ of the form
\[\SYM \circ \XOR \circ \AND \circ \LTF.\] Note that the fan-in of the $\SYM$ gate is at most $2^{n^{a\cdot\eps}}$, and the fan-in of the (merged) $\AND$s is $O(n^{1-\delta/8+a\eps})$.

\textbf{5. Apply modulus-amplifying polynomials to eliminate the XOR layer.} We'd like to remove the $\XOR$ layer, to further reduce the depth of the circuit. But as the gates of this layer have very high fan-in, we must be careful not to blow the circuit size up to $\Omega(2^n)$. The following construction will take advantage of the fact that we have only $\poly(n)$ total gates in the bottom $\LTF$ layer. 

We apply one step of Beigel-Tarui's transformation~\cite{Beigel-Tarui} (from $\ACC^0$ to $\SYM \circ \AND$) to the $\SYM \circ \XOR \circ \AND$ part of our circuit. In particular, we apply a modulus-amplifying polynomial $P$ (over the integers) of degree $2D' = 2n^{a\cdot \eps}$ to each of the $\XOR \circ \AND$ parts. Construing the $\XOR \circ \AND$ as a sum of products $\sum \prod$, the polynomial $P$ has the property:
\begin{compactitem}
\item If the $\sum \prod = 1 \bmod 2$, then $P(\sum \prod) = 1 \bmod 2^{D'}$.
\item If the $\sum \prod = 0 \bmod 2$, then $P(\sum \prod) = 0 \bmod 2^{D'}$.
\end{compactitem}
So, composing $P$ with each $\XOR \circ \AND$ part, each $P$ outputs either $0$ or $1$ modulo $2^{n^{a\cdot \eps}}$. The key property here is that the modulus exceeds the fan-in of the $\SYM$ gate, so the sum of all $P(\sum \prod)$ simply counts the number of $\XOR \circ \AND$s which are true; this is enough to determine the output of the $\SYM$ gate. Construing the output of each bottom $\LTF$ gate as a variable, there are at most $n^{2-\eps}$ variables. Expressing each $P(\sum\prod)$ (expanded as a sum of products) as a multilinear polynomial in these $\LTF$ variables, the total number of terms is at most
\[\binom{n^{2-\eps}}{D' \cdot n^{1-\delta/8+a\eps}} \leq 2^{O(D' \cdot n^{1-\delta/8+a\eps} \cdot \log n)} \leq 2^{O(n^{2a\cdot \eps+ 1-\delta/8} \cdot \log n)}.\] Let $\eps := \delta/(ca)$ for a sufficiently large constant $c > 1$ so that $2a\eps + 1 -\delta/8 < 1-\eps$. We can then merge the sum of all $P(\sum \prod)$'s into the $\SYM$ gate, and obtain a $\SYM \circ \AND$ circuit where the $\SYM$ has fan-in \[2^{O(n^{2a\cdot\eps + (1-\delta/8)})} \leq 2^{O(n^{1-\eps})},\] and the $\AND$ gates have fan-in $O(n^{2a\cdot\eps+(1-\delta/8)}) \leq O(n^{1-\eps})$. The result is a circuit $C_4$ of the form \[\SYM \circ \AND \circ \LTF.\] 

\textbf{6. Replace the bottom threshold gates with DNFs (Theorem~\ref{LTF-to-AND}), and distribute.} 
Note that the circuit $C_4$ has $n-n^{\eps}$ variables, so our SAT algorithm would follow if we could evaluate $C_4$ on all of its variable assignments in $2^{n-n^{\eps}}\cdot \poly(n)$ time. We are now in a position to apply Lemma~\ref{LTF-to-AND}, which lets us reduce the evaluation problem for $\SYM \circ \AND \circ \LTF$ circuits to the evaluation problem for $\SYM \circ \AND \circ \SUM \circ \AND$ circuits, with a parameter $K$ that needs setting. Recall the middle $\AND$ gates have fan-in $O(n^{1-\eps})$, and the fan-in of the $\SUM$ is $O(\log K)$. Therefore by the distributive law, we can rewrite the circuit as a $\SYM \circ \SUM \circ \AND$ circuit, where each $\SUM$ gate has $(\log K)^{O(n^{1-\eps})}$ $\AND$s below it, and at most \emph{one} AND below each $\SUM$ is true. Thus we can wire these $\AND$ gates directly into the top $\SYM$ gate without changing the output.

In more detail, let $A, B = \{0,1\}^{(n-n^{\eps})/2}$, and set $N = 2^{(n-n^{\eps})/2}$ and the integer parameter $K := 2^{b \cdot n^{1-\eps}}$ for a sufficiently large constant $b > 1$. By Lemma~\ref{LTF-to-AND}, we can reduce the SAT problem for $\SYM \circ \AND \circ \LTF$ circuits of size $2^{O(n^{1-\eps})}$ on the set $A \times B = \{0,1\}^{n-n^{\eps}}$ to the SAT problem for $\SYM \circ \SUM \circ \AND$ circuits of size \[2^{O(n^{1-\eps})}\cdot 2^{2b n^{1-\eps}}\cdot n^{2-\delta} \leq 2^{O(n^{1-\eps})}\] on a prescribed set $A' \times B'$ with $|A'|=|B'|=N$ and $A',B' \subseteq \{0,1\}^{2b n^{2-\delta}\cdot n^{1-\eps}}$. By the distributive argument from the previous paragraph, we can convert the $\SYM \circ \SUM \circ \AND$ circuit into a $\SYM \circ \AND$ circuit of size at most 
\[ 2^{O(n^{1-\eps})}\cdot 2^{O(n^{1-\eps}\log \log K)} \leq 2^{O(n^{1-\eps}\log(n))}.\] By Lemma~\ref{LTF-to-AND}, we know that if the $\SYM \circ \AND$ SAT problem is solvable in time $T$ on the set $A' \times B'$, then the SAT problem for $C_4$ on the set $A \times B$ can be solved in time $O\left(T + N^2 \cdot Z/K + N \cdot S\right)\cdot \poly(n)$. 

\textbf{7. Evaluate the depth-two circuit on many pairs of points.}
By applying fast rectangular matrix multiplication in a now-standard way \cite{WilliamsJACM14,WilliamsACCTHR14}, the resulting $\SYM \circ \AND$ circuit of $2^{\tilde{O}(n^{1-\eps})}$ size can be evaluated on all points in $A' \times B'$, in time $\poly(n) \cdot 2^{n-n^{\eps}}$, thus solving its SAT problem. Therefore, the SAT problem for $C_4$ can be solved in time \[\poly(n) \cdot 2^{n-n^{\eps}} + \frac{2^{n-n^{\eps}} \cdot 2^{O(n^{1-\eps})}}{2^{b \cdot n^{1-\eps}}} + 2^{\frac{n-n^{\eps}}{2}}\cdot 2^{O(n^{1-\eps}\log(n))}.\] Setting $b > 1$ to be sufficiently large, we obtain a SAT algorithm for $C_4$ (and hence the original circuit $C$) running in $\poly(n) \cdot 2^{n-n^{\eps}}$ time.
\end{proofof}

\subsection{Satisfiability for Three Layers of Majority + AC0}

In this section, we give our SAT algorithm for $\MAJ \circ \AC^0 \circ \LTF \circ \AC^0 \circ \LTF$  circuits with low-polynomial fan-in at the output gate and the middle $\LTF$ layer:

\begin{reminder}{Theorem~\ref{TC03-SAT}} For all $\eps > 0$ and integers $d \geq 1$, there is a $\delta > 0$ and a randomized satisfiability algorithm for $\MAJ \circ \AC^0 \circ \LTF \circ \AC^0 \circ \LTF$ circuits of depth $d$ running in $2^{n - \Omega(n^{\delta})}$ time, on circuits with the following properties:\begin{compactitem} 
\item the top $\MAJ$ gate, along with every $\LTF$ on the middle layer, has $O(n^{6/5-\eps})$ fan-in, and
\item there are $O(2^{n^{\delta}})$ many $\AND/\OR$ gates (anywhere) and $\LTF$ gates at the bottom layer. 
\end{compactitem} 
\end{reminder}

We need one more result concerning probabilistic polynomials over the integers:

\begin{theorem}[\cite{BeigelRS91,Tarui93}] \label{AC0-integers} For every $\AC^0$ circuit $C$ with $n$ inputs and size $s$, there is a distribution of $n$-variate polynomials ${\cal D}$ over $\Z$ such that every $p$ has degree $\poly(\log s)$ (depending on the depth of $C$) and for all $x \in \{0,1\}^n$, $\Pr_{p \sim {\cal D}}[C(x) = p(x)] \geq 1-1/2^{\poly(\log s)}$.
\end{theorem}

\begin{proofof}{Theorem~\ref{TC03-SAT}} The SAT algorithm is somewhat similar in structure to Theorem~\ref{ACC-LTF-SAT}, but with a few important changes. Most notably, we work with probabilistic polynomials over $\Z$ instead of $\F_2$. 

Start with a circuit $C$ of the required form. Let $s$ be the number of $\AND$/$\OR$ gates in $C$ plus the number of $\LTF$ gates on the bottom layer. Let $f \leq n^{6/5-\eps}$ be the maximum fan-in of the top $\MAJ$ gate and the $\LTF$s on the middle layer, and recall that we're planning to consider $C$ with size at most $2^{n^{\delta}}$ where $\delta > 0$ is a sufficiently small constant (depending on $\eps > 0$ and the circuit depth) in the following. Our SAT algorithm runs as follows:

\begin{enumerate}
\item By Theorem~\ref{LTF-as-AC0MAJ}, every LTF of fan-in $f$ can be replaced by an $\AC^0 \circ \MAJ$ of fan-in $f^{1+o(1)}$ and $\poly(f)$ size. Hence we can reduce $C$ to a circuit of similar size, but of the form \[\MAJ \circ \AC^0 \circ \MAJ \circ \AC^0 \circ \MAJ.\] The fan-ins of the majority gates in the middle and bottom layer can be made at most $n^{6/5 - \eps'}$, for any $\eps' > 0$ which is smaller than $\eps$. To be concrete, let us set $\eps':=\eps/2$.

\item Replace the ``middle'' majority gates of fan-in $n^{6/5-\eps/2}$ with probabilistic polynomials (over $\Z$) of degree $n^{3/5-\eps/4}\poly(\log s)$ and error $1/2^{\poly(\log s)}$~\cite{JoshRyan} (Theorem~\ref{derand} in this paper). Replace all the $\AC^0$ subcircuits of size $s$ by probabilistic polynomials (over $\Z$) of degree $\poly(\log s)$ and error $1/2^{\poly(\log s)}$, via Lemma~\ref{AC0-integers}. Note that the latter $\poly(\log s)$ factor depends on the depth of the circuit.

\item Replace the majority gate at the output (of fan-in $f\leq n^{6/5-\eps}$) with the \emph{probabilistic} PTF of Corollary~\ref{cor:poly}, setting the threshold parameter $s'$ (which is called $s$ in the statement of the corollary) to be $2^{2n^{\delta}}$ and setting the error (called $\eps$ in the statement of the corollary) to be $1/f$. The resulting polynomial has degree $n^{2/5-\eps/3}\cdot \poly(n^{\delta})$.

Applying the distributive law to all the polynomials from steps 2 and 3, the new circuit $C'$ can be viewed as an \emph{integer sum} of at most $T$  $\AND \circ \LTF$ circuits of at most $T$ size, where \[T = 2^{n^{3/5-\eps/4}\cdot n^{2/5-\eps/3}\cdot \poly(\log s, n^{\delta})} = 2^{n^{1-7\eps/12}\cdot \poly(\log s, n^{\delta})}\] and all $\AND$ gates have fan-in at most $n^{1-7\eps/12}\cdot \poly(\log s, n^{\delta})$ (because the resulting polynomial has at most this degree). 

Now is a good time to mention our choice of $\delta$, as it will considerably clean up the exponents in what follows. We will choose $\delta > 0$ to be sufficiently small so that the $\poly(\log s,n^{\delta})$ factor in the exponent of $T$ is less than $n^{\eps/12}$. That is, we take $\delta := \eps/c$ and the size parameter $s < 2^{n^{\delta}}=2^{n^{\eps/c}}$, for a sufficiently large constant $c \geq 12$. (Note that $c$ depends on the depth of the circuit, since the degree of the $\poly\log$ factor depends on the depth.) Thus we have the size bound \[T = 2^{n^{1-7\eps/12}\cdot \poly(\log s, n^{\delta})}\leq O(2^{n^{1-7\eps/12}\cdot n^{\eps/12}})\leq O(2^{n^{1-\eps/2}}),\] and all $\AND$ gates have fan-in at most $n^{1-\eps/2}$.

\item For all assignments $a$ to the first $n^{\delta}$ variables of $C'$, plug $a$ into $C'$, creating a copy $C'_a$. Let $C''$ be the integer sum of all $2^{n^{\delta}}$ circuits $C'_a$. By the properties of the polynomial constructed in Theorem~\ref{thm:poly} and the chosen parameter $s' = 2^{2n^{\delta}}$, with probability at least $2/3$ there is a (computable) threshold value $v = 3s/2$ such that
\begin{itemize}
\item $C''(x) > v$ when at least one $C'_a(x)$ outputs $1$, and 
\item $C''(x) < v$ when all $C'_a(x)$ output $0$. 
\end{itemize}
The circuit $C''$ is a Sum-of-$\AND \circ \LTF$ circuit; note that $C''$ has $n-n^{\delta}$ variables.

\item We now want to evaluate $C''$ on all of its $2^{n-n^{\delta}}$ possible variable assignments. Applying Lemma~\ref{LTF-to-AND} for an integer parameter $K \in [2^n]$ (to be determined), $N=2^{(n-n^{\delta})/2}$, and $Z, S = 2^{n^{1-\eps/2}}$, we can convert this evaluation problem for $C''$ into a corresponding evaluation problem for a Sum-of-$\AND \circ \SUM \circ \AND$ circuit $C'''$, on an appropriate combinatorial rectangle $A' \times B'$ of $2^{n-n^{\delta}}$ variable assignments in total. The relative size of the circuit is unchanged, as each $\SUM \circ \AND$ has size $O(\log^2 K) \leq O(n^2)$. The time for conversion of $C''$ into $C'''$ is \[\left(\frac{N^2 Z^2}{K} + N \cdot S\right)\cdot \poly(n) \leq  \frac{2^{n-n^{\delta}}\cdot 2^{2n^{1-\eps/2}}\cdot \poly(n)}{K}.\] Setting $K := 2^{2 n^{1-\eps/2}}$ makes this time bound $2^{n-\Omega(n^{\delta})}$.

Recall that in the Sum-of-$\AND \circ \SUM \circ \AND$ circuit $C'''$, the fan-in of the middle $\AND$s is at most $n^{1-\eps/2}$, and each $\SUM$ has $O(n)$ fan-in. We can therefore apply the distributive law to each $\AND \circ \SUM$ part, and obtain a $\SUM \circ \AND$ of size at most $n^{O(n^{1-\eps/2})}$. Merging the $\SUM$s into the $\SYM$ gate, we obtain a $\SYM \circ \AND$ circuit of size at most $n^{O(n^{1-\eps/2})}$. 

\item Finally, applying rectangular matrix multiplication (Lemma~\ref{rectangular}) we can evaluate the Sum-of-$\AND$ $C'''$ of $n^{O(n^{1-\eps/2})}$ size on the combinatorial rectangle $A' \times B'$ in $2^{n - \Omega(n^{\delta})}$ time, by preparing matrices of dimensions $2^{n/2 - \Omega(n^{\delta})} \times n^{O(n^{1-\eps/2})}$ (for $A'$) and $n^{O(n^{1-\eps/2})} \times 2^{n/2 - \Omega(n^{\delta})}$ (for $B'$), then multiplying them. Note that preparing these matrices takes time no more than $2^{n/2 + O(n^{1-\eps/2}\log n)}$, which is negligible for us.

After multiplying the matrices, we obtain a value for $C''(x)$ for each assignment $x$, which is correct with probability at least $2/3$. By repeating steps 2-5 for $100 n$ times, we obtain correct values on all $2^{n-n^{\delta}}$ points with high probability.
\end{enumerate}

This completes the proof.
\end{proofof}

\section{Conclusion} 

Our work has led to interesting algorithmic improvements for several core problems. Here are two open problems that we wish to highlight.

First, it would be interesting to understand what are the power and limits of probabilistic polynomial threshold functions representing Boolean functions. How easy/difficult is it to prove degree lower bounds for such representations? In this paper, we have demonstrated how probabilistic PTFs can be significantly better than probabilistic polynomials or deterministic PTFs alone, by combining the strengths of the two representation methods. Informally, a probabilistic polynomial threshold function can be seen as an ${\sf Approximate\text{-}MAJ} \circ \LTF \circ \AND$ circuit or as an ${\sf Approximate\text{-}MAJ} \circ \LTF \circ \XOR$ circuit, so we are effectively asking about lower bounds regarding such circuit classes.

Second, can our SAT algorithm for $\MAJ \circ \AC^0 \circ \LTF \circ \AC^0 \circ \LTF$ be derandomized? If so, the derandomization should lead to new circuit lower bounds. Perhaps the ideas in Tamaki's recent work~\cite{Tamaki16} will be helpful here.

\section*{Acknowledgments} 

The authors thank the FOCS referees for their helpful comments.

\bibliographystyle{alpha}
\bibliography{papers}

\end{document}